\pdfoutput=1

\documentclass{article}
         
\usepackage[english]{babel}
\usepackage[fleqn]{amsmath}
\usepackage{tikz}
\usepackage{amssymb}
\usepackage{lmodern,microtype}
\usepackage[a4paper,left=4cm]{geometry}
\usepackage{hyperref}
\usetikzlibrary{arrows,decorations,calc}
\usetikzlibrary{decorations.pathmorphing,patterns,decorations.pathreplacing,decorations.markings}

\newcommand{\diff}{\text{d}}
\renewcommand{\i}{\text{i}}
\newcommand{\tr}{\mathop{\text{tr}}}

\newtheorem{theorem}{Theorem}[section]
\newtheorem{lemma}[theorem]{Lemma}
\newtheorem{proposition}[theorem]{Proposition}

\newtheorem{conjecture}[theorem]{Conjecture}

\newenvironment{proof}[1][Proof:]{\begin{trivlist}
\item[\hskip \labelsep {\bfseries #1}]}{\end{trivlist}}

\newcommand{\qed}{\nobreak \ifvmode \relax \else
      \ifdim\lastskip<1.5em \hskip-\lastskip
      \hskip1.5em plus0em minus0.5em \fi \nobreak
      $\square$\fi}
      
\tikzset{->-/.style={decoration={
  markings,
  mark=at position #1 with {\arrow{>}}},postaction={decorate}}}
      
\tikzset{
  on each segment/.style={
    decorate,
    decoration={
      show path construction,
      moveto code={},
      lineto code={
        \path [#1]
        (\tikzinputsegmentfirst) -- (\tikzinputsegmentlast);
      },
      curveto code={
        \path [#1] (\tikzinputsegmentfirst)
        .. controls
        (\tikzinputsegmentsupporta) and (\tikzinputsegmentsupportb)
        ..
        (\tikzinputsegmentlast);
      },
      closepath code={
        \path [#1]
        (\tikzinputsegmentfirst) -- (\tikzinputsegmentlast);
      },
    },
  },
  mid arrow/.style={postaction={decorate,decoration={
        markings,
        mark=at position .625 with {\arrow[#1]{stealth}}
      }}},
}

\newcommand{\drawvertex}[3]
{
 
  \ifnum#1=1
   \draw[xshift=#2 cm,yshift=#3 cm,postaction={on each segment={mid arrow}}] (0,0)--
(.5,0)--(1,0);
  \draw[xshift=#2 cm,yshift=#3 cm,postaction={on each segment={mid arrow}}] (.5,-.
5)--(.5,0)--(.5,.5);
  \fi
   \ifnum#1=2
      \draw[xshift=#2 cm,yshift=#3 cm,postaction={on each segment={mid arrow}}] 
(1,0)--(.5,0)--(0,0);
  \draw[xshift=#2 cm,yshift=#3 cm,postaction={on each segment={mid arrow}}] (.
5,.5)--(.5,0)--(.5,-.5); 
  \fi
\ifnum#1=3
      \draw[xshift=#2 cm,yshift=#3 cm,postaction={on each segment={mid arrow}}] 
(1,0)--(.5,0)--(0,0);
  \draw[xshift=#2 cm,yshift=#3 cm,postaction={on each segment={mid arrow}}] (.5,-.
5)--(.5,0)--(.5,.5);   \fi
\ifnum#1=4
      \draw[xshift=#2 cm,yshift=#3 cm,postaction={on each segment={mid arrow}}] 
(0,0)--(.5,0)--(1,0);
  \draw[xshift=#2 cm,yshift=#3 cm,postaction={on each segment={mid arrow}}] (.
5,.5)--(.5,0)--(.5,-.5);  
  \fi
\ifnum#1=5
      \draw[xshift=#2 cm,yshift=#3 cm,postaction={on each segment={mid arrow}}] 
(0,0)--(.5,0)--(.5,.5);
  \draw[xshift=#2 cm,yshift=#3 cm,postaction={on each segment={mid arrow}}] (1,0)--
(.5,0)--(.5,-.5);  
  \fi
\ifnum#1=6
      \draw[xshift=#2 cm,yshift=#3 cm,postaction={on each segment={mid arrow}}] (.
5,.5)--(.5,0)--(0,0);
  \draw[xshift=#2 cm,yshift=#3 cm,postaction={on each segment={mid arrow}}] (.5,-.
5)--(.5,0)--(1,0);  
  \fi
}

\title{\large \bf The nineteen-vertex model and alternating sign matrices}
\date{}

\author{\normalsize  \sc{Christian Hagendorf}\medskip
\\
\normalsize \it
Universit\'e Catholique de Louvain\\
\normalsize \it Institut de Recherche en Math\'ematique et Physique\\
\normalsize \it Chemin du Cyclotron 2, 1348 Louvain-la-Neuve, Belgium
\medskip\\
\href{mailto:christian.hagendorf@uclouvain.be}{\normalsize 
\texttt{christian.hagendorf@uclouvain.be}}
}

\begin{document}
\maketitle
\begin{abstract}
It is shown that the transfer matrix of the inhomogeneous nineteen-vertex model with 
certain diagonal twisted boundary conditions possesses a simple eigenvalue. This is 
achieved through the identification of a simple and completely explicit solution of 
its Bethe equations.
The corresponding eigenvector is computed by means of the algebraic Bethe ansatz, 
and both a simple component and its square norm are expressed in terms of the Izergin-Korepin determinant. In 
the homogeneous limit, the vector coincides with a supersymmetry singlet of the 
twisted spin-one XXZ chain. It is shown that in a natural polynomial normalisation 
scheme its square norm and the simple component coincide with generating functions for weighted enumeration of alternating sign matrices.
\end{abstract}

\section{Introduction}
The study of integrable quantum spin chains relies often on the solution of the 
Bethe equations \cite{baxterbook,gaudin:83}. These are a set of non-linear coupled 
equations for the so-called Bethe roots, involving rational, trigonometric or even 
elliptic functions depending to the case at hand. Their solutions allow to build 
eigenstates of the Hamiltonian, and are thus in principle the starting point for the 
calculation of physically relevant quantities such as correlation functions. 
However, solving these equations is in general a challenging problem. Especially in 
finite-size systems one has to recourse frequently to numerical solutions because 
the patterns of Bethe roots for the ground state of the system and low-lying excited 
states are typically very complicated. Simplifications occur in infinite systems 
where the Bethe equations can be transformed into integral equations for 
which various analytical solution methods exist.

It is therefore to be expected that non-trivial integrable models for which the 
exact finite-size Bethe roots for the ground state can be computed (explicitly) by analytical 
methods are rather scarce. One of the few examples is the spin$-1/2$ XXZ chain with 
anisotropy $\Delta=-1/2$. Indeed, the $\mathcal Q$-function, a polynomial whose 
roots coincide with the Bethe roots, was found exactly for the ground state of 
finite chains with and various boundary conditions: periodic, twisted and open 
\cite{fridkin:00,fridkin:00_2}. Extensions to the spin$-1/2$ XYZ along a particular 
line of couplings are known, and the corresponding $\mathcal Q$-function displays 
even remarkable relations to classically integrable equations \cite{bazhanov:05,bazhanov:06}. It seems however that in each of these cases the boundary 
conditions have to be fine-tuned: twist angles or boundary magnetic fields need to 
be adjusted to particular values, the system length has to be even or odd depending 
on the particular choice of boundary conditions, etc.

The purpose of this article is to present an example of a quantum integrable model 
for which the Bethe roots of a highly non-trivial eigenstate of the Hamiltonian can be exactly and explicitly determined in 
finite size. It is the spin-one XXZ chain \cite{zamolodchikov:81,fateev:81} with 
particular, fine-tuned twisted boundary conditions but \textit{arbitrary 
anisotropy}. In \cite{hagendorf:13}, this spin chain was shown to possess a 
supersymmetric structure on the lattice in certain (anti-)cyclic subsectors of its Hilbert space. 
The ground states in these sectors are the so-called supersymmetry singlets. We show 
that at least one of these singlets can be obtained from the Bethe ansatz when all 
the Bethe roots coincide. This case is known to be difficult to handle as it has to 
be defined through a suitable limiting procedure \cite{baxter:02} (at least within the framework of the coordinate Bethe ansatz). In order to 
circumvent this technical obstacle, we use the common trick to introduce 
inhomogeneities into the model in such a way that its integrability is preserved. 
While the notion of a local spin-chain Hamiltonian and lattice supersymmetry are 
absent in the inhomogeneous case, the transfer matrix of the corresponding nineteen-vertex
model remains a well-defined object to study. We 
identify a special boundary condition with a fine-tuned twist angle for which it 
possesses a simple eigenvalue with a corresponding eigenstate whose Bethe roots are 
shown to coincide simply with the inhomogeneity parameters.

The existence of explicit Bethe roots and a simple eigenvalue does of course not 
necessarily imply that the corresponding eigenstate is interesting. Yet, a look at 
the literature on the above-mentioned spin$-1/2$ XXZ and XYZ chains, and the vertex-
models which they are related to, shows that these states are in fact the ground 
states of the spin chains, and that in suitable normalisation their components 
display remarkable connections to problems of enumerative combinatorics, most 
importantly the enumeration of alternating sign matrices and plane partitions 
\cite{razumov:00,razumov:01,batchelor:01,razumov:10,mangazeev:10}. A fruitful 
approach to proving these properties is indeed the introduction of inhomogeneity 
parameters which allowed to analyse the eigenvectors in terms of the so-called 
quantum Knizhnik-Zamolodchikov system. Its polynomial solutions allow to determine 
various components, sum rules, and even exact finite-size correlation functions in 
the inhomogeneous case, and then take the homogeneous limit, see for example 
\cite{difrancesco:05_3,difrancesco:06,zinnjustin:13,cantini:12_1}. We show here that 
a similar rich structure can be found in the transfer-matrix eigenstate of the 
twisted inhomogeneous nineteen-vertex model corresponding to the simple eigenvalue. 
Furthermore, we present relations to problems of weighted enumerations of 
alternating sign matrices in the homogeneous limit. Our approach is based on the 
explicit construction of the eigenstate by means of the algebraic Bethe ansatz, and 
the analysis of its properties through known results on scalar products such as 
Slavnov's formula \cite{korepin:93,slavnov:89}. We show that with a suitable non-trivial normalisation convention a certain simple component and the square norm of the supersymmetry singlet 
coincide with generating functions for the weighted enumeration of alternating sign 
matrices.

The layout of this article is the following. We start in section 
\ref{sec:defspinchain} with a discussion of the quantum spin-one XXZ chain, recall 
briefly its lattice supersymmetry and state our results about its supersymmetry 
singlets. Section \ref{sec:19v} is a reminder on the construction of the
nineteen-vertex model from the fusion procedure. We prove the existence of a simple 
eigenvalue of the transfer matrix and thus of the spin-chain Hamiltonian in section 
\ref{sec:aba}. Starting from elementary properties of the corresponding eigenvector 
in section \ref{sec:eigenvector} we prove a relation between the square norm of the 
inhomogeneous eigenstate and the so-called Izergin-Korepin determinant, and use it to find a closed expression for a particular component of the vector. The evaluation 
of their homogeneous limit yields a relation between the 
supersymmetry singlet and alternating sign matrices. We present our conclusions in 
section \ref{sec:concl}.

\section{The spin-one XXZ chain}
\label{sec:defspinchain}

The purpose of this section is to recall the definition of the spin-one XXZ chain with diagonal 
twisted boundary conditions as well as its supersymmetry properties, and then 
state our results about a particular supersymmetry singlet and its relation to the 
enumeration of alternating sign matrices.

\paragraph{Hilbert space and spin operators.} The Hilbert space of the quantum spin chain with $N$ 
sites is given by
\begin{equation}
  V = V_1\otimes V_2 \otimes \cdots \otimes V_N,
  \label{eqn:hilbertspace}
\end{equation}
where every factor is a copy of the Hilbert space for a single spin one, $V_j \simeq \mathbb 
C^3$. We label the canonical basis vectors as follows
\begin{equation*}
  |{\Uparrow}\rangle =
  \left(
  \begin{array}{c}
  1\\ 0 \\ 0
  \end{array}
  \right),\quad  |0\rangle=
  \left(
  \begin{array}{c}
  0\\ 1 \\ 0
  \end{array}
  \right),\quad
   |{\Downarrow}\rangle=
  \left(
  \begin{array}{c}
  0\\ 0 \\ 1
  \end{array}
  \right).
\end{equation*}
The most simple choice of a basis of $V$ is the set of vectors $|\sigma_1,\dots,
\sigma_N\rangle = \bigotimes_{j=1}^N |\sigma_j\rangle$ where $\sigma_j =\, \Uparrow,
0,\Downarrow$ for all $j=1,\dots,N$.

The spins themselves are described by the spin-one representation of $\mathfrak{su}
(2)$ which is given by
\begin{equation*}
   s^{1} = \frac{1}{\sqrt{2}}
   \left(
   \begin{array}{ccc}
   0 & 1 & 0\\
   1 & 0 & 1\\
   0 & 1 & 0
   \end{array}
   \right), \quad
   s^{2} = \frac{1}{\sqrt{2}}
   \left(
   \begin{array}{ccc}
   0 & -\i & 0\\
   \i & 0 & -\i\\
   0 & \i & 0
   \end{array}
   \right),\quad
   s^{3} = 
   \left(
   \begin{array}{ccc}
   1 & 0 & 0\\
   0 & 0 & 0\\
   0 & 0 & -1
   \end{array}
   \right).
\end{equation*}
As usual, we denote by $s_j^a$ the operator $s^a$ acting on the $j$-th factor of the 
tensor product \eqref{eqn:hilbertspace}. The total magnetisation is given by the 
operator $\Sigma = \sum_{j=1}^N s_j^3$. It is diagonal in the canonical basis. 
Moreover, since we are going to consider periodic systems, it will be useful to 
introduce a shift operator which translates the system by one site $S:V_1\otimes V_2 
\otimes \cdots \otimes V_N\mapsto V_N\otimes V_1 \otimes \cdots \otimes V_{N-1}$, 
and hence transforms the spins according to $S s^a_j S^{-1} = s^a_{j+1}$ for 
$j=1,\dots,N-1$, and $S s^a_N S^{-1} = s^a_{1}$.

\paragraph{Hamiltonian.} The Hamiltonian of the spin chain considered in this 
article is given by \cite{zamolodchikov:81,fateev:81}
\begin{subequations}
\begin{equation}
  H = \sum_{j=1}^N \left(\sum_{a=1}^3 J_a (s_j^a s_{j+1}^a+ 2(s_j^a)^2) - 
\sum_{a,b=1}^3 A_{ab}s_j^as_j^b s_{j+1}^a s_{j+1}^b\right),
\end{equation}
where $A$ is a symmetric matrix $A_{ab}=A_{ba}$ with diagonal elements $A_{aa}=J_a$. 
The remaining constants depend only on a single parameter $x$ which measures the 
anisotropy of the spin chain. They are given by
\begin{equation}
  J_1 = J_2 =1,\, J_3 = \frac{1}{2}(x^2-2),\quad  A_{12}=1,\, A_{13}=A_{23}=x-1.
\end{equation}
 \label{eqn:spin1XXZ}%
\end{subequations}
The Hamiltonian can be derived from an integrable vertex model which results from 
fusion of the six-vertex model as we shall see below. It is called the integrable spin-one XXZ 
chain as the derivation is similar to the way the standard spin$-1/2$ XXZ chain can 
be obtained from the six-vertex model.
Let us mention some special cases for which the Hamiltonian simplifies. At $x=2$, 
the Hamiltonian describes the $SU(2)$-symmetric Babujian-Takhtajan spin 
chain \cite{babujian:82,babujian:83}. The point $x=0$ is closely related to the so-called 
supersymmetric $t-J$ model \cite{bares:90}. In the limit $x\to \infty$ the 
spin chain becomes Ising-like, and hence very easy to analyse. Except for this last 
case, the diagonalisation of the Hamiltonian is a non-trivial problem. Nonetheless 
it can actually be done by using the Bethe ansatz.

In order to characterise the spin chain completely, we need to specify its boundary 
conditions. In this article, we are going to investigate the following, so-called 
diagonal twisted boundary conditions:
\begin{align*}
  s_{N+1}^1 = \cos \phi\, s_1^1-\sin\phi\, s_1^2, \quad
s_{N+1}^2 = \sin \phi\, s_1^1+\cos\phi\, s_1^2,\quad \quad s_{N+1}^3=s_{1}^3.
\end{align*}
In fact, this corresponds to a simple rotation of the spin around axis 3 by the twist angle $\phi$ 
when going from site $N$ to site $1$. If $\phi=0$, the boundary conditions are 
periodic and lead to a translation invariance of the spin chain $[H,S]=0$. For non-zero twist angles, the system is however not invariant under translations. Yet, it 
is possible to introduce an appropriate notion of translation invariance by 
considering the modified translation operator $S'= S \Omega_N$. Here $\Omega_N$ is 
an operator acting on the last site, before the system is translated by one site. 
One verifies that $[H,S']=0$ for the boundary conditions given above, provided that
\begin{align}
  & \Omega=
  \left(
  \begin{array}{ccc}
    e^{\i\phi} & 0 & 0\\
    0 & 1 & 0\\
    0 & 0 & e^{-\i \phi}
  \end{array}
  \right).
\label{eqn:defOmega}
\end{align}
Eventually, it is easy to see that for any twist angle the Hamiltonian commutes with 
the total magnetisation along the third axis, $[H,\Sigma]=0$.

\paragraph{Supersymmetric sectors and singlets.} It can be shown that the 
Hamiltonian \eqref{eqn:spin1XXZ} with the twists \eqref{eqn:defOmega} has an exact 
lattice supersymmetry in certain subsectors of the Hilbert space. This means that 
there is an operator $Q$ with $Q^2=0$ such that $H$ can be written as the 
anticommutator
\begin{equation*}
  H=\{Q,Q^\dagger\}.
\end{equation*}
The special feature of $Q$ in the present case is that it increases the number of 
sites by one whereas $Q^\dagger$ decreases the length of the chain by one. The 
supersymmetry is therefore dynamic in the sense that the length of the chain 
changes through the action of its supercharges. The precise definition of $Q$ and further technical can be found in \cite{hagendorf:13}, and 
a recent, very concise and general discussion of dynamic lattice supersymmetry 
in (super)spin chains in \cite{meidinger:13}. For our purposes, it is sufficient to know that the subsectors of $V$ where the supersymmetry exists are the eigenspaces of the twisted translation operator $S'$ with 
eigenvalue $(-1)^{N+1}$. Hence they depend on the twist angle. Indeed, the given 
eigenvalue implies that $(S')^N$ needs to act like the identity which leads to the 
condition that every vector in the subsector has to be an eigenvector of the operator $
\Omega \otimes \Omega \otimes \cdots \otimes \Omega$ with eigenvalue one. Using the 
explicit form \eqref{eqn:defOmega} we find that this is possible if and only if $
\phi \Sigma$ is an integer multiple of $2\pi$. For example, this holds for periodic 
boundary conditions $\phi=0$, or for the case $\phi=\pi$ provided that one restricts
to subsectors where $\Sigma$ is an even integer.

The existence of a supersymmetric structure on certain subspaces of $V$ implies that 
within them the Hamiltonian is a positive definite operator. Its eigenvalues/energy levels 
are bounded from below by zero. If a state $|\Phi\rangle$ with $H|\Phi\rangle=0$ 
exists it is therefore automatically a ground state of the Hamiltonian in these 
subsectors. Such states are called supersymmetry singlets or simply zero-energy states, and 
annihilated by both the supercharge and its adjoint:
\begin{equation*}
  Q|\Phi\rangle = 0,\, Q^\dagger |\Phi\rangle =0.
\end{equation*}
The existence of such a supersymmetry singlet for \eqref{eqn:spin1XXZ} on chains of 
arbitrary length $N$, and arbitrary $x$ was observed for chains of small length with 
twist angle $\phi = \pi$ in \cite{frahm:11_2,hagendorf:13}.
Here, we prove this statement:
\begin{theorem}
For any $N>1$ and twist angle $\phi = \pi$ the Hamiltonian possesses a zero-energy state with 
zero total magnetisation in the subsector of the Hilbert space where the lattice 
supersymmetry exists.
\label{thm:existence}
\end{theorem}
The proof relies on an explicit construction of the eigenstate.
It is however important to stress that this might not be the only singlet. While this 
appears to be the case for most values of $x$, numerical studies of small systems suggests that 
there are special values for $x$, for instance $x=0$, where 
additional zero-energy states occur. We address the uniqueness problem in a more general setting in conjecture \ref{conj:uniqueness}. Moreover, we do not claim neither that the singlet is also the ground state 
when taking into account the full Hilbert space $V$. The exact diagonalisation of 
the Hamiltonian for small systems suggests that this might only be the case for large enough\footnote{The author would like to thank Robert Weston and Junye Yang for pointing this out to him.}  $x$ but a proof of this statement is beyond the scope of this article. 

\paragraph{From supersymmetry to combinatorics.} How does the zero-energy state look like? It is clear from the form of the 
Hamiltonian that we may choose its normalisation such that it is a polynomial in $x$ 
with non-zero constant term. We expand it in the canonical basis according to
\begin{equation*}
  |\Phi(x)\rangle = \sum_{\sigma \in \{\Uparrow,0,\Downarrow\}^N} 
\Phi_{\sigma_1\cdots\sigma_N}(x)|\sigma_1,\dots,\sigma_N\rangle.
\end{equation*}
The components $\Phi_{\sigma_1\cdots\sigma_N}(x)$ are all polynomials in $x$. For example, in the case of $N=3$ sites 
we obtain:
\begin{align}
  &\Phi_{\Uparrow 0 \Downarrow}(x)=\Phi_{\Downarrow 0 \Uparrow}(x)=  1,\nonumber\\
  & \Phi_{\Downarrow\Uparrow 0}(x) = \Phi_{0 \Downarrow\Uparrow}(x) = \Phi_{\Uparrow\Downarrow 0}(x) = \Phi_{0 \Uparrow\Downarrow}(x) =-1,\label{eqn:phiN3} \\
 &  \Phi_{000}(x)= x. \nonumber
\end{align}
In \cite{hagendorf:13} it was 
observed that some of these components, 
and in particular its square norm are given by generating functions for a certain 
type of weighted enumeration of alternating sign matrices. These are matrices 
with entries $-1,0,1$, and the rules that along each row and column the non-zero 
elements alternate in sign, with the first and last non-zero entry being $1$. For 
instance, all $3\times 3$ alternating sign matrices are given by
\begin{equation*}
\begin{array}{ccccccc}
\left(
\begin{smallmatrix}1&0&0\\0&1&0\\0&0&1\end{smallmatrix}
\right),
&
\left(
\begin{smallmatrix}1&0&0\\0&0&1\\0&1&0\end{smallmatrix}
\right),
&
\left(
\begin{smallmatrix}0&1&0\\1&0&0\\0&0&1\end{smallmatrix}
\right),
&
\left(
\begin{smallmatrix}0&1&0\\0&0&1\\1&0&0\end{smallmatrix}
\right),
&
\left(
\begin{smallmatrix}0&0&1\\1&0&0\\0&1&0\end{smallmatrix}
\right),
&
\left(
\begin{smallmatrix}0&0&1\\0&1&0\\1&0&0\end{smallmatrix}
\right),
&
\left(
\begin{smallmatrix}0&1&0\\1&\!-1\!&1\\0&1&0\end{smallmatrix}
\right).
\end{array}
\end{equation*}
A closed formula for the number of $N\times N$ alternating sign matrices was 
conjectured by Mills, Robbins and Rumsey, and later proved by Zeilberger 
\cite{zeilberger:96} (see \cite{bressoudbook} for an overview). A short proof 
borrowing methods from quantum integrability was subsequently found by Kuperberg 
\cite{kuperberg:96}. The enumeration problem can be refined as follows 
\cite{kuperberg:02}: assign a weight $t^k$ to all to every alternating matrix with 
exactly $k$ entries $-1$, and sum these weights for all $N\times N$ matrices. The 
result $A_N(t)$ is obviously a polynomial in $t$. For $N=3$ we see from the matrices 
shown here above that $A_3(t) = 6 + t$. An explicit expression for $A_N(t)$ with 
arbitrary $N$ in terms of a determinant formula can be found in \cite{robbins:00,behrend:12}. It implies in particular that $A_N(t)$ is a polynomial of degree $\lfloor (N-1)^2/4\rfloor$.

Let us compare the example for $N=3$ to the square norm of the vector 
\eqref{eqn:phiN3}. We use the \textit{real} scalar product $\langle \cdot |
\cdot \rangle$ on $V$, and therefore compute the norm $||\Phi(x)||^2 = \langle 
\Phi(x)|\Phi(x)\rangle$ according to
\begin{equation*}
  ||\Phi(x)||^2=\sum_{\sigma \in \{\Uparrow,0,\Downarrow\}^N} \Phi_{\sigma_1\cdots
\sigma_N}(x)^2.
\end{equation*}
For our example $N=3$, we find
\begin{equation*}
  ||\Phi(x)||^2 = 6 + x^2 = A_3(t=x^2).
\end{equation*}
This is not only coincidence for three sites. One verifies by exact diagonalisation 
of the Hamiltonian for small $N$ that in a suitable normalisation, where all components are polynomials in $x$ with integer coefficients, the 
square norm of $|\Phi(x)\rangle$ is indeed equal to $A_N(t=x^2)$.

To make this more precise, we have to fix the degree of the state as a polynomial in $x$ in order to avoid redundancies. It is clear that we may restrict ourselves to impose the degree of a specific component. We choose $\Phi_{\Uparrow\cdots\Uparrow\Downarrow\cdots\Downarrow}(x)$ for even $N=2n$, and $\Phi_{\Uparrow\cdots\Uparrow0\Downarrow\cdots\Downarrow}(x)$ for odd $N=2n+1$, and require them to be polynomials in $x$ of degree $\lfloor (n-1)^2/4\rfloor$ with the constant term being adjusted to
\begin{equation*}
      \Phi_{ \underset{n}{\underbrace{\Uparrow\cdots\Uparrow}}\underset{n}{\underbrace{\Downarrow\cdots\Downarrow}}}(x)= n!+O(x), 
\quad 
      \Phi_{\underset{n}{ \underbrace{\Uparrow\cdots\Uparrow}}0\underset{n}{\underbrace{\Downarrow\cdots\Downarrow}}}(x)= n!+O(x). 
\end{equation*}
With this convention it is no longer possible to multiply the state with arbitrary polynomials which would only generate common (and redundant) factors of the components. We claim that with this normalisation scheme all components of the singlet are polynomials in $x$ with integer coefficients. Moreover, the two special components are then given by
  \begin{equation}
    \Phi_{ \underset{n}{\underbrace{\Uparrow\cdots\Uparrow}}\underset{n}{\underbrace{\Downarrow\cdots\Downarrow}}}(x) = \Phi_{\underset{n}{ \underbrace{\Uparrow\cdots\Uparrow}}0\underset{n}{\underbrace{\Downarrow\cdots\Downarrow}}}(x) = A_n(x^2),
    \label{eqn:simplecomponents}
  \end{equation}
 and the square norm of the singlet takes the form
  \begin{equation}
    ||\Phi(x)||^2 = A_N(x^2).
    \label{eqn:sumrule}
  \end{equation}
  Here $A_N(t)$ is the generating function for the weighted enumeration of $N\times N$ alternating sign matrices with weight $t$ per entry $-1$.
  
The normalisation convention presented is different from the one in \cite{hagendorf:13} where a restriction on the leading coefficient of $|\Phi(x)\rangle$ as a 
polynomial in $x$ was imposed. The two conditions appear to be equivalent. However, the one stated here seems to be easier to prove.

\section{The nineteen-vertex model}
\label{sec:19v}
The spin-chain Hamiltonian \eqref{eqn:spin1XXZ} can be related to an integrable 
vertex model, the so-called nineteen-vertex model. The relation allows to study the 
spin chain with the help of tools from quantum integrability such as the algebraic Bethe 
ansatz. This is indeed the strategy which we pursue in order to prove the existence 
of the supersymmetry singlet. To this end, we need to recall the construction of 
the nineteen-vertex model through the so-called fusion procedure, and introduce 
furthermore the transfer matrices of the model with twisted boundary conditions and 
inhomogeneities. 

\paragraph{Fusion.} We use a parametrisation in terms of multiplicative spectral 
parameters, and make systematic use of the following abbreviation
\begin{equation*}
[z]=z-z^{-1}.
\end{equation*}
The fusion procedure \cite{kulish:81,kulish:82} allows to construct iteratively 
solutions of the Yang-Baxter equation $R^{(m,n)}(z)\in \text{End}(\mathbb C^{m+1} 
\otimes \mathbb C^{n+1}),\,m,n=1,2,\dots$, starting from $m=n=1$: they solve
\begin{equation}
  R_{12}^{(m,n)}(z/w)R_{13}^{(m,p)}(z) R_{23}^{(n,p)}(w)=
   R_{23}^{(n,p)}(w)R_{13}^{(m,p)}(z)R_{12}^{(m,n)}(z/w),
  \label{eqn:ybe}
\end{equation}
on the product space $\mathbb C^{m+1} \otimes \mathbb C^{n+1}\otimes \mathbb C^{p+1}
$. The indices $i,j$ of $R_{ij}^{(m,n)}(z)$ label the factors of the tensor product 
which the $R$-matrices act on. The simplest case $m=n=1$ corresponds $R$-matrix of 
the six-vertex model. Let us abbreviate the canonical basis of $\mathbb C^2$ by
\begin{equation*}
  |{\uparrow}\rangle =
  \left(
  \begin{array}{c}
  1\\0
  \end{array}
  \right), \quad |{\downarrow}\rangle =
  \left(
  \begin{array}{c}
  0\\1
  \end{array}
  \right).
\end{equation*}
Then, in the basis $\{|{\uparrow\uparrow}\rangle, |{\uparrow\downarrow}\rangle, |
{\downarrow\uparrow}\rangle, |{\downarrow\downarrow}\rangle\}$ of $\mathbb C^2 
\otimes \mathbb C^2$, we have
\begin{equation*}
  R^{(1,1)}(z) =
  \left(
    \begin{array}{cccc}
    [qz] & 0 & 0 & 0 \\
    0 & [z] & [q] & 0\\
    0 & [q] & [z] & 0\\
    0 & 0 & 0 & [qz]
    \end{array}
  \right),
\end{equation*}
which solves \eqref{eqn:ybe} with $m=n=p=1$. The matrix $R^{(1,1)}(z)$ degenerates 
at $z=q$ and $z=q^{-1}$ where it can be written in terms of the projectors $P^+$ and 
$P^-$ onto the symmetric and antisymmetric subspaces of $\mathbb C^2\otimes \mathbb 
C^2$:
\begin{equation*}
  R^{(1,1)}(z=q) = B P^+,\quad R^{(1,1)}(z=q^{-1}) = (-2[q]) P^-,
\end{equation*}
with the diagonal matrix $B = \text{diag}([q^2],2[q],2[q],[q^2])$. This simple 
observation allows to construct $R^{(1,2)}(z)$ from the evaluation of the Yang-Baxter
equation at the degeneration points. For example, setting $w= q^{-1}$ in the 
Yang-Baxter equation we see that $P_{23}^+R_{12}^{(1,1)}(qz)R_{13}^{(1,1)}(z)$ 
leaves stable the symmetric subspace of the second and third factor in the tensor 
product. This can be extended to the following decomposition of a product of $R$-matrices:
\begin{equation*}
  M_{23} R_{12}^{(1,1)}(qz)R^{(1,1)}_{13}(z)M_{23}^{-1} =
  \left(
  \begin{array}{cc}
    [qz]R_{1,(23)}^{(1,2)}(z) & 0\\
    \ast & [z][q^2 z]
    \end{array}
  \right).
\end{equation*}
The rows and columns of the matrix on the right-hand side are indexed by the symmetric 
and antisymmetric subspaces of $\mathbb C^2\otimes \mathbb C^2 = \text{Sym}^2 
\mathbb C^2 \oplus \bigwedge^2 \mathbb C^2$. 
Moreover, $M = \text{diag}(1/\sqrt{[q^2]},1/\sqrt{2[q]},1/\sqrt{2[q]},
1/\sqrt{[q^2]})$ is a diagonal matrix, whose introduction  leads to a symmetric 
matrix $R^{(1,2)}(z)$. If we identify the basis vectors of $\text{Sym}^2 \mathbb C^2 
$ with the basis vectors of $\mathbb C^3$ according to 
\begin{equation*}
  |{\Uparrow}\rangle = |{\uparrow\uparrow}\rangle,\quad |0\rangle = \frac{1}
{\sqrt{2}}(|{\uparrow\downarrow}\rangle + |{\downarrow\uparrow}\rangle),\quad |
{\Downarrow}\rangle = |{\downarrow\downarrow}\rangle,
\end{equation*}
then the $R^{(1,2)}(z)$ can be evaluated in compact form. Indeed, one shows that in 
the basis $\{|{\uparrow{\Uparrow}}\rangle,|{\uparrow 0}\rangle,|
{\uparrow{\Downarrow}}\rangle,|{\downarrow{\Uparrow}}\rangle,|{\downarrow 
0}\rangle,|{\downarrow {\Downarrow}}\rangle\}$ it is given by
\begin{equation}
  R^{(1,2)}(z)
  =\left(
  \begin{array}{cccccc}
  [q^2 z] & 0 & 0 & 0 & 0 & 0\\
  0 & [q z] & 0 & \sqrt{[q][q^2]} & 0 & 0\\
  0 & 0 & [z] & 0 & \sqrt{[q][q^2]} & 0\\
  0 & \sqrt{[q][q^2]} & 0 & [z] & 0 & 0\\
  0 & 0 & \sqrt{[q][q^2]} & 0 & [q z] & 0 \\
  0 & 0 & 0 & 0 & 0 & [q^2 z]
  \end{array}
  \right).
  \label{eqn:r12}
\end{equation}
One verifies that it solves the Yang-Baxter equation \eqref{eqn:ybe} with 
$m=n=1,\,p=2$. In fact, it is known that $R^{(1,m)}(z)$ for all $m=1,2,\dots$ can be 
written down systematically in terms of the generators of the quantum group $U_{q}
(\mathfrak{sl}_2)$ \cite{kirillov:87,yung:95}. The non-zero matrix elements of 
$R^{(1,2)}(q^{-1}z)$ can be interpreted as statistical weights for a vertex model on 
the square lattice whose horizontal edges are always oriented, whereas the vertical 
edges may be either oriented or not. The structure of \eqref{eqn:r12} 
leaves us with the configurations of a mixed vertex model, shown in figure 
\ref{fig:verticesR12}.
\begin{figure}[h]
  \centering
  \begin{tikzpicture}
     \draw[postaction={on each segment={mid arrow}}] (0,0) -- (.5,0) -- (1,0);
     \draw[postaction={on each segment={mid arrow}},double] (.5,-.5) -- (.5,0) --(.
5,.5);
     \begin{scope}[xshift=1.5cm]
    \draw[postaction={on each segment={mid arrow}}] (1,0) -- (.5,0) -- (0,0);
     \draw[postaction={on each segment={mid arrow}},double] (.5,.5) -- (.5,0) -- (.
5,-.5) ;
    \end{scope}
    
    \begin{scope}[xshift=3.5cm]
       \draw[postaction={on each segment={mid arrow}}] (0,0) -- (.5,0) -- (1,0);
     \draw[postaction={on each segment={mid arrow}},double]  (.5,.5) -- (.5,0) --(.
5,-.5);
    \end{scope}
    \begin{scope}[xshift=5cm]
    \draw[postaction={on each segment={mid arrow}}] (1,0) -- (.5,0) -- (0,0);
     \draw[postaction={on each segment={mid arrow}},double] (.5,-.5)  -- (.5,0) --  
(.5,.5);

    \end{scope}
    
     \begin{scope}[xshift=7cm]
     \draw[postaction={on each segment={mid arrow}}] (0,0) -- (.5,0) -- (1,0);
     \draw[densely dotted ,double]  (.5,.5)  --(.5,-.5);
    \end{scope}
    \begin{scope}[xshift=8.5cm]
    \draw[postaction={on each segment={mid arrow}}] (1,0) -- (.5,0) -- (0,0);
     \draw[densely dotted,double] (.5,-.5)  -- (.5,0) --  (.5,.5);

    \end{scope}
    
    \begin{scope}[yshift=-2.5cm,xshift=2.cm]
     \draw[postaction={on each segment={mid arrow}}] (.5,0) -- (0,0);
     \draw[postaction={on each segment={mid arrow}}] (.5,0) -- (1,0);  
     \draw[densely dotted,double] (.5,0)  -- (.5,.5);
     \draw[postaction={on each segment={mid arrow}},double] (.5,-.5) -- (.5,0);  
     
     \begin{scope}[xshift=1.5cm]
     \draw[postaction={on each segment={mid arrow}}] (0,0) -- (.5,0);
     \draw[postaction={on each segment={mid arrow}}] (1,0) -- (.5,0);  
     \draw[densely dotted,double] (.5,0)  -- (.5,.5);
     \draw[postaction={on each segment={mid arrow}},double] (.5,0) -- (.5,-.5);  
     \end{scope}
     
     \begin{scope}[xshift=3cm]
     \draw[postaction={on each segment={mid arrow}}]   (.5,0)-- (0,0) ;
     \draw[postaction={on each segment={mid arrow}}]   (.5,0)-- (1,0);  
     \draw[densely dotted,double] (.5,0)  -- (.5,-.5);
     \draw[postaction={on each segment={mid arrow}},double] (.5,.5) -- (.5,0);  
     \end{scope}
     
     \begin{scope}[xshift=4.5cm]
    \draw[postaction={on each segment={mid arrow}}] (0,0) -- (.5,0);
     \draw[postaction={on each segment={mid arrow}}] (1,0) -- (.5,0); 
     \draw[densely dotted,double] (.5,0)  -- (.5,-.5);
     \draw[postaction={on each segment={mid arrow}},double] (.5,0) -- (.5,.5);  
     \end{scope}

    \end{scope}
    \draw [decoration={brace,mirror,raise=0.7cm},decorate] (0,0) -- (2.5,0); 
    \draw [decoration={brace,mirror,raise=0.7cm},decorate,xshift=3.5cm] (0,0) -- 
(2.5,0); 
    \draw [decoration={brace,mirror,raise=0.7cm},decorate,xshift=7cm] (0,0) -- 
(2.5,0); 
    
    \draw (1.25,-1.1) node {$[qz]$};
    \draw [xshift=3.5cm] (1.25,-1.1) node {$[q^{-1}z]$};
    \draw [xshift=7cm] (1.25,-1.1) node {$[z]$};
    
    \draw [decoration={brace,mirror,raise=0.7cm},decorate] (2,-2.5) -- (7.5,-2.5); 
    \draw (4.75,-3.6) node {$\sqrt{[q][q^2]}$};
  \end{tikzpicture}   
  \caption{Vertices with non-zero weights of the mixed vertex model derived from the 
matrix $R^{(1,2)}(q^{-1}z)$. The spin components $\uparrow$ and $\downarrow$ are 
represented by simple arrows along the horizontal direction, oriented to the right 
and left respectively. Along the vertical direction, the spin components $\Uparrow, 
0, \Downarrow$ are represented by upwards oriented lines, dotted lines without 
orientation or downwards oriented lines. In order to reconstruct the matrix 
elements, the pictures have to be read from south-west to north-east by following 
the lines. For example, from the first vertex in the bottom row we find $
\langle{\uparrow} 0|R^{(1,2)}(q^{-1} z)|{\downarrow\Uparrow}\rangle = \sqrt{[q]
[q^2]}$.}
  \label{fig:verticesR12}
\end{figure}
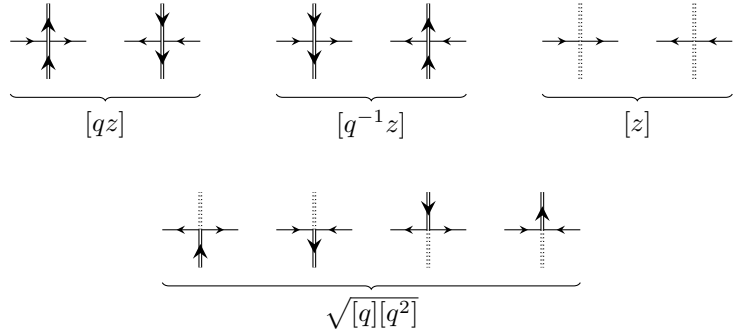

The $R$-matrix of the nineteen-vertex model is constructed by applying once more the 
decomposition into symmetric, and antisymmetric subspaces to a product of $R$-matrices $R^{(1,2)}(z)$:
\begin{equation}
  M_{12}R_{23}^{(1,2)}(z)R_{13}^{(1,2)}(q^{-1}z)^{(1,2)}M_{12}^{-1} = 
\left(\begin{array}{cc}
    R_{(12),3}^{(2,2)}(z) & 0\\
    \ast & [q^{-1}z][q^2 z]
    \end{array}
  \right).
  \label{eqn:fusionR12}
\end{equation}
The non-zero matrix elements of $R(z)=R^{(2,2)}(z)$ correspond to the weights of the 
configurations of the integrable nineteen-vertex model which are shown in figure 
\ref{fig:19vertexweights}. 
Each edge of the lattice is either oriented (with two possibilities 
for the orientation) or without orientation. The configurations of the model are constrained by the rule 
that at the number of edges oriented towards any vertex is equal to the number of 
edges oriented away from it.
\begin{figure}[h]
  \centering
  \begin{tikzpicture}[>=stealth]
   
    \draw[postaction={on each segment={mid arrow}},double] (0,0) to (0.5,0) to (.
5,.5);
    \draw[postaction={on each segment={mid arrow}}, double] (.5,-.5) to (.5,0) to 
(1,0);
    
    \begin{scope}[xshift=1.5cm]
    \draw[postaction={on each segment={mid arrow}},double] (.5,.5) to (0.5,0) to 
(0,0);
    \draw[postaction={on each segment={mid arrow}},double] (1,0) to (.5,0)  to (.
5,-.5);
    \end{scope}
    
    \begin{scope}[xshift=3.5cm]
    \draw[postaction={on each segment={mid arrow}},double] (0,0) to (.5,0) to (.5,-.
5);
    \draw[postaction={on each segment={mid arrow}},double]  (.5,.5) to (.5,0) to 
(1,0);
    \end{scope}
    
    \begin{scope}[xshift=5cm]
   \draw[postaction={on each segment={mid arrow}},double] (.5,-.5) to (.5,0) to 
(0,0) ;
    \draw[postaction={on each segment={mid arrow}},double]  (1,0) to (.5,0)to (.
5,.5) ;    \end{scope}
    
     \begin{scope}[xshift=7cm]
    \draw[postaction={on each segment={mid arrow}},double] (0,0) to (0.5,0) to (.
5,.5);
    \draw[postaction={on each segment={mid arrow}},double] (1,0) to (.5,0)  to (.
5,-.5);
    \end{scope}
    
     \begin{scope}[xshift=8.5cm]
    \draw[postaction={on each segment={mid arrow}},double] (.5,.5) to (0.5,0) to 
(0,0);
    \draw[postaction={on each segment={mid arrow}},double] (.5,-.5) to (.5,0) to 
(1,0);
    \end{scope}
    
    \begin{scope}[yshift=-2.5cm]
        \draw[densely dotted,double] (1,0) to (0,0);
        \draw[postaction={on each segment={mid arrow}},double] (.5,-.5) to (.5,0) to 
(.5,.5);

      \begin{scope}[xshift=1.5cm]
        \draw[densely dotted,double] (0,0) to (1,0);
        \draw[postaction={on each segment={mid arrow}},double] (.5,.5) to (.5,0) to 
(.5,-.5);
      \end{scope}
      
      \begin{scope}[xshift=3cm]
        \draw[densely dotted,double] (.5,.5) to (.5,-.5);
        \draw[postaction={on each segment={mid arrow}},double] (0,0) to (.5,0) to  
(1,0);
      \end{scope}
      
      \begin{scope}[xshift=4.5cm]
        \draw[densely dotted,double] (.5,-.5) to (.5,.5);
        \draw[postaction={on each segment={mid arrow}},double] (1,0) to (.5,0) to  
(0,0);
      \end{scope}
     
      \begin{scope}[xshift=6.5cm]
        \draw[densely dotted,double] (0,0) to (0.5,0) to (.5,.5);
        \draw[postaction={on each segment={mid arrow}},double] (.5,-.5) to (.5,0) to 
(1,0);
      \end{scope}
      
      \begin{scope}[xshift=8cm]
        \draw[densely dotted,double] (.5,.5) to (0.5,0) to (0,0);
        \draw[postaction={on each segment={mid arrow}},double] (1,0) to (.5,0)  to 
(.5,-.5);
      \end{scope}

      \begin{scope}[xshift=9.5cm]
        \draw[densely dotted,double] (1,0) to (.5,0)  to (.5,-.5);
        \draw[postaction={on each segment={mid arrow}},double] (.5,.5) to (0.5,0) to 
(0,0);
      \end{scope}
     
      \begin{scope}[xshift=11cm]
        \draw[densely dotted,double] (.5,-.5) to (.5,0) to (1,0);
        \draw[postaction={on each segment={mid arrow}},double] (0,0) to (0.5,0) to 
(.5,.5);
      \end{scope}
    \end{scope}
    
    \begin{scope}[yshift=-5cm]
      \draw[densely dotted,double] (.5,.5) to (0.5,0) to (1,0);
      \draw[postaction={on each segment={mid arrow}},double] (.5,-.5) to (.5,0)  to 
(0,0);         
 
      \begin{scope}[xshift=1.5cm]
        \draw[densely dotted,double] (.5,.5) to (0.5,0) to (1,0);
        \draw[postaction={on each segment={mid arrow}},double] (0,0) to (.5,0)  to 
(.5,-.5);
      \end{scope}
      
      \begin{scope}[xshift=3cm]
        \draw[densely dotted,double] (0,0) to (.5,0) to (.5,-.5);
        \draw[postaction={on each segment={mid arrow}},double] (.5,.5) to (.5,0) to 
(1,0);
      \end{scope}
      
      \begin{scope}[xshift=4.5cm]
        \draw[densely dotted,double]  (.5,-.5) to (.5,0) to (0,0);
        \draw[postaction={on each segment={mid arrow}},double] (1,0) to  (.5,0) to 
(.5,.5);
      \end{scope}

      \begin{scope}[xshift=8.75cm]
        \draw[densely dotted,double] (0,0) to (0.5,0) to (.5,.5);
        \draw[densely dotted,double] (.5,-.5) to (.5,0) to (1,0);
      \end{scope}
    \end{scope}

    \draw [decoration={brace,mirror,raise=0.7cm},decorate] (0,0) -- (2.5,0); 
    \draw [decoration={brace,mirror,raise=0.7cm},decorate,xshift=3.5cm] (0,0) -- 
(2.5,0); 
    \draw [decoration={brace,mirror,raise=0.7cm},decorate,xshift=7cm] (0,0) -- 
(2.5,0); 
    
    \draw (1.25,-1.1) node {$[qz][q^2 z]$};
    \draw [xshift=3.5cm] (1.25,-1.1) node {$[q^{-1}z][z]$};
    \draw [xshift=7cm] (1.25,-1.1) node {$[q][q^2]$};
    
    \draw [decoration={brace,mirror,raise=0.7cm},decorate] (0,-2.5) -- (5.5,-2.5); 
    \draw [decoration={brace,mirror,raise=0.7cm},decorate,xshift=6.5cm] (0,-2.5) -- 
(5.5,-2.5); 
    
    \draw (2.75,-3.6) node {$[z][qz]$};
    \draw (9.25,-3.6) node {$[q^2][qz]$};
    
    \draw [thick,decoration={brace,mirror,raise=0.7cm},decorate] (0,-5) -- (5.5,-5);     
    \draw (2.75,-6.1) node {$[q^2][z]$};
    \draw (9.25,-6.1) node {$[z][qz]+[q][q^2]$};
  
  \end{tikzpicture}
  \caption{Vertices and weights of the nineteen-vertex model: the components $
\Uparrow$ are represented by edges oriented to the right or top, $0$ by 
dotted double edges without orientation, $\Downarrow$ by edges pointing down or to the left. In order to 
reconstruct the matrix elements of $R(z)=R^{(2,2)}(z)$ one reads the configurations 
from south-west to north-east by following the lines. For example : $\langle 
0{\Uparrow}|R(z)|{\Uparrow}0\rangle=[q^2][qz]$.}
  \label{fig:19vertexweights}
\end{figure}
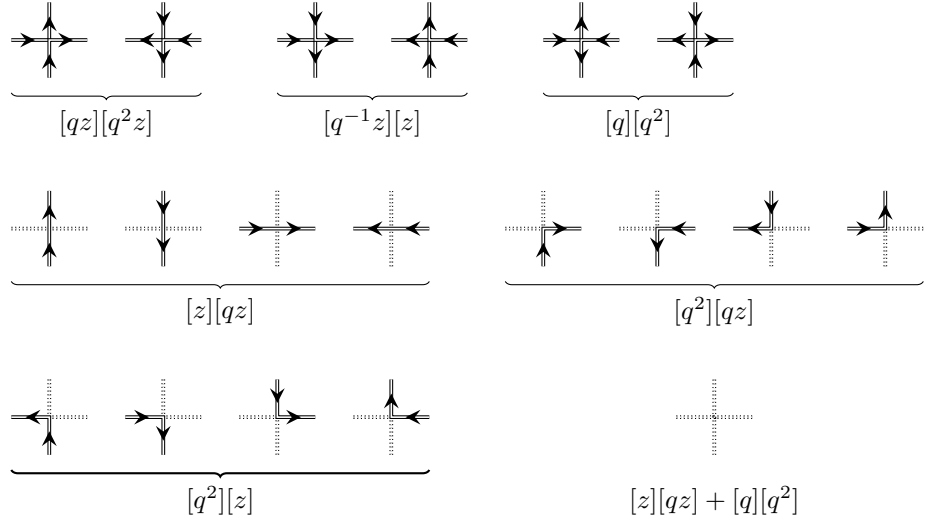

Before proceeding to the construction of the transfer matrices of the nineteen-vertex model, we mention a few properties of the $R$-matrices. One verifies that the 
$R$-matrices $R^{(1,1)}(z),\,R^{(1,2)}(z)$ and $R(z)=R^{(2,2)}(z)$ solve the Yang-
Baxter equation \eqref{eqn:ybe} for all admissible combinations. 
The $R$-matrix of the nineteen-vertex model is of special interest to us. If we set 
its argument to  $z=1$ then it reduces up to a factor to the permutation operator $P
$ on $\mathbb C^3 \otimes \mathbb C^3$, i.e. the operator defined through $P(|
v_1\rangle \otimes |v_2\rangle)= |v_2\rangle \otimes |v_1\rangle$:
\begin{equation}
  R(z=1) = [q][q^2]P.
  \label{eqn:permutation}
\end{equation}
The $R$-matrix is symmetric and therefore we have $P R(z) P = R(z)$. It is 
occasionally useful to use the abbreviation $\check R(z) = PR(z) = R(z)P$ which 
reduces up to a factor to the identity at $z=1$. Furthermore, one verifies by direct 
calculation that the so-called inversion relation holds
\begin{equation*}
  R(z)R(z^{-1}) = r(z)r(z^{-1}), \quad r(z) = [q/z][q^2z].
\end{equation*}
The fact that the left-hand side vanishes for certain values of $z$ means that $R(z)
$ cannot be invertible for all $z$. Like the six-vertex $R$-matrix $R^{(1,1)}(z)$ it 
has some degeneration points where it reduces to projectors on certain subspaces of 
$\mathbb C^3 \otimes \mathbb C^3$. A particular interesting point is $z=q^{-1}$ 
where its rank is one. We have
\begin{equation}
  R(z=q^{-1}) = [q][q^2]|s\rangle\langle s|,\quad |s\rangle = |{\Uparrow\Downarrow}
\rangle + |{\Downarrow\Uparrow\rangle}-|{00}\rangle.
  \label{eqn:rankone}
\end{equation}

\paragraph{Transfer matrices.} Let us consider a transfer matrix built from the $R$-matrix
$R^{(1,2)}(z)$:
\begin{equation}
  T^{(1)}(z) = {\tr}_{a} \left(\Omega_a^{(1)} R_{a,N}^{(1,2)}(q^{-1}z/w_j)\cdots 
R_{a,2}^{(1,2)}(q^{-1}z/w_2)R_{a,1}^{(1,2)}(q^{-1}z/w_1)\right).
  \label{eqn:defT12}
\end{equation}
Here $w_j,\,j=1,\dots, N$ is an inhomogeneity parameter attached to the $j$-th 
factor of the $N$-fold tensor product $V$. If they all take the same value then the 
model is called homogeneous, otherwise inhomogeneous. The trace is taken over an 
auxiliary (horizontal) space $\mathbb C^2$, labeled by $a$. The twist $\Omega^{(1)}$ 
denotes an operator acting on this auxiliary space. In this article, we consider the 
following diagonal case
\begin{equation}
  \Omega^{(1)} =
  \left(
    \begin{array}{cc}
    e^{\i \phi/2} & 0\\
    0 & e^{-\i \phi/2}
    \end{array}
    \right).
  \label{eqn:twists}
\end{equation}
This choice yields a phase according to the direction of the spin on the auxiliary 
space before taking the trace.

The fusion procedure outlined above allows to construct the transfer matrix of the 
nineteen-vertex model $T^{(2)}(z)$ from $T^{(1)}(z)$. It is defined as
\begin{equation*}
  T^{(2)}(z) = {\tr}_{a} \left(\Omega_a^{(2)}R_{a,N}(z/w_j)\cdots R_{a,2}(z/
w_2)R_{a,1}(z/w_1)\right),
\end{equation*}
where the auxiliary space labeled by $a$ is now $\mathbb C^3$. With the twist from
\eqref{eqn:twists}, the relation between the two types of transfer matrices is given 
by
\begin{equation}
  T^{(2)}(z) = T^{(1)}(z)T^{(1)}(q z)+(-1)^{N+1}\prod_{j=1}^N [qw_j/z][q^2z/w_j].
  \label{eqn:fusion}
\end{equation}
The derivation of this equation is based on the decomposition of the double trace 
over the two auxiliary spaces in the product $T^{(1)}(z)T^{(1)}(q z)$ into separate 
traces over the symmetric and antisymmetric subspaces of the product space, followed 
by the use of \eqref{eqn:fusionR12}. $\Omega^{(2)}$ is determined from $\Omega^{(1)}
$ by
\begin{equation*}
  \Omega^{(2)} = P^+ (\Omega^{(1)} \otimes \Omega^{(1)}) P^+
  =
  \left(
  \begin{array}{ccc}
    e^{\i\phi} & 0 & 0\\
    0 & 1 & 0\\
    0 & 0 & e^{-\i \phi}
  \end{array}
  \right),
\end{equation*}
where $P^+$ is the projector onto the symmetric subspace of $\mathbb C^2\otimes \mathbb C^2$. The result twist is therefore the same as in \eqref{eqn:defOmega}.

Let us now remind a few properties of the transfer matrices which are going to be 
relevant to our considerations. We emphasise here below the dependence on the 
inhomogeneity parameters $T^{(j)}(z) = T^{(j)}(z|w_1,\dots,w_N)$ for 
$j=1,2$. One of their most important features is that 
for a given set of inhomogeneities these matrices commute for different values of 
the spectral parameter:
\begin{equation*}
  [T^{(j)}(z|w_1,\dots,w_N),T^{(k)}(z'|w_1,\dots,w_N)] =0,
\end{equation*}
for all $z,z'$ and all choices for $j,k=1,2$.
The commutation relations imply in particular that 
the eigenvectors of these matrices are independent of the spectral parameter $z$. Furthermore, 
it is sufficient to diagonalise $T^{(1)}(z|w_1,\dots,w_N)$ in order to compute 
both the spectrum and the eigenvectors of the $T^{(2)}(z|w_1,\dots,w_N)$.
Another well-known consequence of the commutation relations is that it implies the 
existence of a large number of commuting operators which contain at the homogeneous point $w_1=\dots=w_N=1$ both the twisted translation operator and the Hamiltonian of 
the spin-one XXZ chain \eqref{eqn:spin1XXZ}. Indeed, using \eqref{eqn:permutation} 
one finds at the point $z=1$:
\begin{subequations}
\begin{align}
  S' = ([q][q^2])^{-N}T^{(2)}(z=1|w_1=1,\dots, w_N=1).
\end{align}
The Hamiltonian is obtained from the logarithmic derivative at the same point
\begin{align}
  H = N+ \frac{[q^2]}{2}\left.\frac{\diff}{\diff z} \ln T^{(2)}(z|w_1=1,\dots, 
w_N=1)\right|_{z=1},
\end{align}
provided that the parameter $x$ used in section \ref{sec:defspinchain} is identified 
with
\begin{equation}
  x= q+q^{-1}.
\end{equation}
\label{eqn:SHT}%
\end{subequations}
Hence, the diagonalisation of the transfer matrices allows to recover the 
eigenvalues and -vectors of the spin-chain Hamiltonian.

\section{Simple eigenvalue}
\label{sec:aba}
In this section, we show that the inhomogeneous transfer matrix of the nineteen-vertex model with a diagonal twist possesses a simple eigenvalue for the twist angle 
$\phi=\pi$. The observation is astonishingly simple: we show that the Bethe 
equations of the model admit an explicit solution in terms of the inhomogeneity 
parameters. Taking the homogeneous limit, we observe that it leads to a zero 
eigenvalue of the Hamiltonian.

\paragraph{Algebraic Bethe ansatz.} In the case of diagonal twists, the eigenvalues 
and -vectors of $T^{(1)}(z)$ and $T^{(2)}(z)$ can be constructed from the algebraic 
Bethe ansatz \cite{korepin:93,faddeev:96}. The basic object is the monodromy matrix 
of the inhomogeneous model for $N$ sites, given by
\begin{equation*}
  \mathcal T_a(z) = R_{a,N}^{(1,2)}(q^{-1}z/w_N)\cdots R_{a,1}^{(1,2)}(q^{-1}z/w_1).
\end{equation*}
It can be seen as a $2\times 2$ matrix acting on the auxiliary space whose entries 
are operators on the Hilbert space $V$:
\begin{equation*}
  \mathcal T_a(z) = \left(
  \begin{array}{cc}
  \mathcal A(z) & \mathcal B(z)\\
  \mathcal C(z) & \mathcal D(z)
  \end{array}
  \right)_a.
\end{equation*}
Comparing with \eqref{eqn:defT12} we conclude that
\begin{align*}
  T^{(1)}(z|w_1,\dots,w_N) &= {\tr}_{a} \Omega_a^{(1)}\mathcal T_a(z)= e^{\i \phi/2} 
\mathcal A(z) + e^{-\i \phi/2} \mathcal D(z).
\end{align*}  
The entries of the monodromy matrix satisfy a number of quadratic relations 
which are a result of $R_{a_1a_2}^{(1,1)}(z/w) \mathcal T_{a_1}(z) \mathcal T_{a_2}
(w) = \mathcal T_{a_2}(w)\mathcal T_{a_1}(z) R_{a_1a_2}^{(1,1)}(z/w)$, an immediate 
consequence of the Yang-Baxter equation \eqref{eqn:ybe}. Here, we quote only a few 
of them relevant to our discussion:
\begin{subequations} 
\begin{align}
  & \mathcal A(w)\mathcal B(z) = f(w,z)\mathcal B(z)\mathcal A(w)+ g(w,z) \mathcal 
B(w)\mathcal A(z),\\
  & \mathcal D(w)\mathcal B(z) = f(z,w)\mathcal B(z)\mathcal D(w)+ g(z,w) \mathcal 
B(w) \mathcal D(z),\\
  & [\mathcal B(z), \mathcal B(w)] = [\mathcal C(z),\mathcal C(w)]=0,
\end{align}
\label{eqn:commadb}%
where we abbreviate
\begin{equation*}
  f(z,w)=[q w/z]/[w/z], \quad \text{and}\quad  g(z,w) = [q]/[w/z].
\end{equation*}%
\end{subequations}

The algebraic Bethe ansatz is based on a reference state (pseudo vacuum) $|{\wedge}
\rangle$ defined through $\mathcal C(z)|{\wedge}\rangle=0$ for any $z$. From the 
explicit form of $R^{(1,2)}(z)$ it is not difficult to see that it is given by the 
completely polarised state
\begin{equation*}
  |{\wedge}\rangle = |{\Uparrow\Uparrow}\, \cdots\, {\Uparrow\Uparrow}\rangle.
\end{equation*}
Furthermore, inspecting \eqref{eqn:r12} one concludes that the action of the 
diagonal elements of the monodromy matrix on this state is very simple:
\begin{align*}
 \mathcal A(z)|{\wedge}\rangle = a(z)|{\wedge}\rangle, \quad a(z) = \prod_{j=1}^N 
[q z/w_j],\quad  \mathcal D(z)|{\wedge}\rangle = d(z)|{\wedge}\rangle, \quad d(z) 
= \prod_{j=1}^N [q^{-1}z/w_j].
\end{align*}
The eigenstates of $T^{(1)}(z)$ are constructed by acting with the operator $
\mathcal B(z)$ on the reference state. It is easy to see that $[\Sigma,\mathcal 
B(z)]=-\mathcal B(z)$: this operator flips spins and lowers the total magnetisation 
by one. An eigenvector with magnetisation $N-n$ is given by
\begin{equation}
  |\Psi(z_1,\dots,z_n)\rangle = \prod_{j=1}^n \mathcal B(z_j)|{\wedge}\rangle.
  \label{eqn:bethestate}
\end{equation}
As any two $\mathcal B$-operators commute according to \eqref{eqn:commadb} the 
expression is obviously symmetric in the $z_1,\dots,z_n$. These numbers are however 
not arbitrary. One shows that they lead to an eigenvector of $T^{(1)}(z)$  with 
eigenvalue
\begin{equation}
  \theta^{(1)}(z|w_1,\dots, w_n) = e^{\i \phi/2}a(z)\prod_{j=1}^n f(z,z_j) +e^{-\i 
\phi/2}d(z)\prod_{j=1}^n f(z_j,z)\label{eqn:betheeigenvalue},
\end{equation}
provided that the so-called Bethe equations hold. They are given by $a(z_k)/
d(z_k)=e^{-\i \phi} \prod_{j \neq k}^n f(z_j,z_k)/f(z_k,z_j)$, or more explicitly:
\begin{equation}
  \prod_{j=1}^N \frac{[q z_k/w_j]}{[q^{-1}z_k/w_j]} = e^{-\i \phi}\prod_{j\neq k}^n
\frac{[q z_k/z_j]}{[q^{-1}z_k/z_j]}, \quad k=1,\dots,n.
  \label{eqn:bae}
\end{equation}
The derivation of these statements is based on the repeated application of the 
quadratic relations \eqref{eqn:commadb} in order to evaluate $\mathcal A(z)|
\Psi(z_1,\dots,z_n)\rangle$ and $\mathcal D(z)|\Psi(z_1,\dots,z_n)\rangle$ 
\cite{korepin:93}. The Bethe equations result from the elimination of so-called 
unwanted terms which occur when commuting $\mathcal A(z),\mathcal D(z)$ through the 
operators $\mathcal B(z_j)$ in \eqref{eqn:bethestate}. The eigenvalue is obtained 
after this when acting with $\mathcal A(z),\mathcal D(z)$ on the reference state. 
Because of the relation between $T^{(1)}(z)$ and $T^{(2)}(z)$ this procedure leads 
automatically to an eigenvector of the transfer matrix for the nineteen-vertex 
model. The eigenvalue can be computed from \eqref{eqn:fusion}. In the homogeneous 
limit $w_1=\cdots=w_N=1$ it allows therefore to diagonalise the spin-chain 
Hamiltonian.

\paragraph{Simple eigenvalue.} Explicit solutions of the coupled algebraic equations 
\eqref{eqn:bae} with finite $N$ and $n$ are rather difficult to obtain. Yet, in the 
present case there is a remarkably simple solution, provided that the twist angle is 
chosen to be $\phi = \pi$ and the number of Bethe roots coincides with the number of 
sites $N=n$. It is given by
\begin{equation}
  z_k =  w_k, \quad k=1,\dots, N.
  \label{eqn:betheroots}
\end{equation}
Its insertion into \eqref{eqn:eigenvalue} leads to a vanishing eigenvalue $
\theta^{(1)}(z|z_1=w_1,\dots, z_N=w_N) = 0$ for all $z$. We apply the fusion 
equation \eqref{eqn:fusion}, and obtain the following result:
\begin{theorem}
  For the twist angle $\phi=\pi$ and any number of sites $N$, the transfer matrix of 
the nineteen-vertex model $T^{(2)}(z)$ possesses the eigenvalue
  \begin{equation}
    \theta^{(2)}(z|w_1,\dots,w_N) = (-1)^{N+1}\prod_{j=1}^N [q w_j/z][q^2 z/w_j].
    \label{eqn:eigenvalue}
  \end{equation}
\end{theorem}
This is in fact an inhomogeneous generalisation of the existence statements of 
theorem \ref{thm:existence}, appearing now as a mere corollary. Indeed, 
\eqref{eqn:eigenvalue} holds for any choice of the inhomogeneity parameters and 
therefore by continuity also for the homogeneous model. We conclude therefore that if 
$w_1=\dots= w_N=1$ the transfer matrix possesses a Bethe eigenvector whose Bethe 
roots condense to a single point $z_1=\dots=z_N=1$. Using the relation between the 
transfer matrix and the spin-chain Hamiltonian given in \eqref{eqn:SHT} we conclude 
that the corresponding eigenstate has zero energy, and is an eigenvector of the 
twisted translation operator with eigenvalue $(-1)^{N+1}$. Hence, it is a ground 
state of the Hamiltonian in the subsector of the Hilbert space where the 
supersymmetry exists. Moreover, from $n=N$ we conclude that the resulting state has 
zero magnetisation. This completes the proof of theorem \ref{thm:existence}.

As for the homogeneous case discussed in section \ref{sec:defspinchain}, we can at this point however not conclude that the eigenvalue non-degenerate even though the numerical diagonalisation for small system sizes suggests that this is the case:
\begin{conjecture}
  Except for possibly a finite number of values of the parameter $q$, the eigenvalue \eqref{eqn:eigenvalue} is non-degenerate for all $N$.
  \label{conj:uniqueness}
\end{conjecture}

\section{Properties of the eigenvector}

\label{sec:eigenvector}
In this section, we investigate the properties of the eigenvector of the transfer 
matrices $T^{(j)}(z|w_1,\dots,w_N),\,j=1,2$ with the simple Bethe roots 
\eqref{eqn:betheroots}. It is given by
\begin{equation}
  |\Psi(w_1,\dots,w_N)\rangle = \prod_{j=1}^N \mathcal B(w_j)|{\wedge}\rangle.
  \label{eqn:eigenvector}
\end{equation}
We start with analysing some simple properties of the state in section 
\ref{sec:elementary}. In section \ref{sec:scalar} we use some of them in combination 
with Slavnov's formula in order to prove a multi-parameter sum rule for the 
eigenvector. We apply these results in section \ref{sec:simplecomponents} in order
to determine a simple component of the vector. The homogeneous point $w_1=\cdots=w_N=1$ is studied in section 
\ref{sec:homogeneous}.

\subsection{Elementary properties}
\label{sec:elementary}
Here below, we present a simple graphical interpretation of the Bethe vector 
\eqref{eqn:eigenvector}. Next, we establish an exchange relation which allows to 
permute its inhomogeneity parameters through the action of $R$-matrices. We use it 
in order to show that the Bethe vector contains quite a few redundant factors. Hence 
we introduce a renormalised vector, and show that its solves a set of relations 
which are similar to the quantum Knizhnik-Zamolodchikov system. Eventually we 
determine its degree width as a Laurent polynomial in each variable.

\paragraph{Ground-state components as partition functions of the mixed vertex 
model.} The Bethe eigenvector \eqref{eqn:eigenvector} can be depicted graphically. 
In order to build it, we start the completely polarised state and act with the $
\mathcal B$-operators $N$ times. From the definition of the monodromy matrix we 
infer that these operators can be depicted as
\begin{equation}
\begin{tikzpicture}[baseline=.35]
  \draw (-3,0) node {$\mathcal B(z) = $};
    \draw (-1.5,0) -- (-.5,0);
    \draw (0.5,0) -- (1,0);
    \draw[dotted] (-0.5,0) -- (.5,0);

    \draw[postaction={on each segment={mid arrow}}] (-1.5,0) -- (-2,0);
    \draw[postaction={on each segment={mid arrow}}] (1,0) -- (1.5,0);
    \draw[double] (-1.5,-.5)--(-1.5,.5);
    \draw[double] (-1.,-.5)--(-1.,.5); 
    \draw[double] (1.,-.5)--(1.,.5);   
      
    \draw (0,-.5) node[below] {$\cdots$};
    \draw (-1.5,-.5) node[below] {\small $w_1$};
    \draw (-1.,-.5) node[below] {\small $w_2$};
    \draw (1.,-.5) node[below] {\small $w_N$};
    \draw(-2,0) node[above] {\small $z$};
  \end{tikzpicture}
  \label{eqn:Bgraphical}
\end{equation}
where we indicated the parameters attached to the horizontal and vertical lines. The 
$N$-fold action is obtained from the vertical concatenation of these pictures which 
results in an $N\times N$ square grid whose horizontal lines have outpointing 
arrows, whereas the vertical lines along the bottom row have arrows inwards. The 
projection onto a basis vector fixes the boundary condition along the top row of the 
square, and leads to the following graphical representation for the component $
\Psi_{\sigma_1\sigma_2\cdots\sigma_N}(w_1,\dots,w_N)=\langle \sigma_1,\dots,
\sigma_N|\Psi(w_1,\dots,w_N)\rangle$:
\begin{equation}
  \begin{tikzpicture}[baseline=-.15cm]
  \draw (-5,0) node {$\Psi_{\sigma_1\sigma_2\cdots\sigma_N}(w_1,\dots,w_N)=$};
  \foreach \y in {-1.5,-1,...,1}
  {
    \draw (-1.5,\y) -- (1,\y);
    \draw[postaction={on each segment={mid arrow}}] (-1.5,\y) -- (-2,\y);
    \draw[postaction={on each segment={mid arrow}}] (1,\y) -- (1.5,\y);
    \draw[double] (\y,-1.5) -- (\y,1.5);
    \draw[double,postaction={on each segment={mid arrow}}] (\y,-2) -- (\y,-1.5);  
  }
       \draw (-1.5,1.5) node[above] {$\sigma_1$};
       \draw (-1.,1.5) node[above] {$\sigma_2$};
       \draw (1.,1.5) node[above] {$\sigma_N$};
       \draw (0,1.5) node[above] {$\cdots$};
       \draw (0,-2) node[below] {$\cdots$};
       \draw (-1.5,-2) node[below] {\small $w_1$};
       \draw (-1.,-2) node[below] {\small $w_2$};
       \draw (1.,-2) node[below] {\small $w_N$};
       
       \draw (-2,-1.5) node[left] {\small $w_1$};
       \draw (-2,-1) node[left] {\small $w_2$};
       \draw (-2,1) node[left] {\small $w_N$};
       \draw (-2.2,0) node[left] {$\vdots$};
  \end{tikzpicture}
  \label{eqn:graphicalrepr}
\end{equation}
Any component of the Bethe vector coincides therefore with a partition function of the 
mixed vertex model on such an $N\times N$ square. The Boltzmann weight for a vertex in row $i$ and column $j$ is given by $R^{(1,2)}
(q^{-1}w_i/w_j)$, the weight of a configuration by the product of its vertex weights, and the partition function by the sum over the weights of all configurations compatible with the boundary conditions. Because of the 
commutativity of the $\mathcal B$-operators for different arguments 
\eqref{eqn:commadb}, the vertical arrangement of the parameters $w_1,\dots,w_N$ is completely arbitrary. The graphical representation will be
useful in order to compute a certain simple component of the eigenvector, and study the point $w_1=\dots=w_N=1$.

\paragraph{Exchange relation.} Let us now analyse how to permute the arguments of the Bethe vector
\eqref{eqn:eigenvector}. Notice that even though a generic Bethe state 
\eqref{eqn:bethestate} is symmetric in the Bethe roots we cannot conclude that 
$|\Psi(w_1,\dots,w_N)\rangle$ is symmetric in its arguments, because the $\mathcal B$-operators 
carry an explicit dependence \textit{and} implicit on these parameters. 
Nonetheless, it is possible to derive an exchange relation. To this end, we examine 
the action of the matrix $\check R(z)$ on the vector \eqref{eqn:eigenvector}.
It is useful to study  the commutation relation between the monodromy matrix and the $R$-matrix 
of the nineteen-vertex model. Writing $\mathcal T_a(z) = \mathcal T_a(z|
w_1,\dots,w_N)$ we obtain\begin{equation*}
 \mathcal T_a(z|\dots,w_{j+1},w_{j},\dots)\check R_{j,j+1}(w_j/w_{j+1})=\check 
R_{j,j+1}(w_j/w_{j+1})
\mathcal T_a(z|\dots,w_{j},w_{j+1},\dots),
\end{equation*}
for $j=1,\dots,N-1$.
Indeed, this is easily proved by writing both sides as a product of $R$-matrices, 
and applying the Yang-Baxter equation. As $\check R_{j,j+1}(w_j/w_{j+1})$ does not 
act on the auxiliary space, this relation holds for all matrix elements of 
the monodromy matrix, in particular for $\mathcal B(z)=\mathcal B(z|w_1,\dots,w_N)$. 
Taking into 
account that $\check R_{j,j+1}(z)|\wedge\rangle = [qz][q^2 z]|\wedge\rangle
$ we find thus from \eqref{eqn:eigenvector} the relation
\begin{align}
  \check R_{j,j+1}\left(\frac{w_j}{w_{j+1}}\right)|\Psi(\dots, w_j,w_{j+1},\dots)
\rangle = \left[\frac{qw_j}{w_{j+1}}\right]\left[\frac{q^2w_j}{w_{j+1}}\right]|
\Psi(\dots, w_{j+1},w_{j},\dots)\rangle.
  \label{eqn:exchangeeigenvector}
\end{align}

\paragraph{Renormalised eigenvector.} We show now that 
\eqref{eqn:exchangeeigenvector} implies that the eigenvector is proportional to a 
product of elementary factors which are common to all its components: 
\begin{proposition}
All components of the vector $|\Psi(w_1,\dots,w_N)\rangle$ are proportional to the 
product $\prod_{1\leq j < k \leq N} [qw_j/w_k]$.
\end{proposition}
\begin{proof}
We evaluate the action of the transfer matrix $T^{(2)}(w=w_j|w_1,\dots,w_N)$ on 
the vector in two ways. On the one hand, we may use the fact that 
\eqref{eqn:eigenvector} is an eigenvector with eigenvalue \eqref{eqn:eigenvalue}. 
On the other hand, at $w=w_j$ the 
transfer matrix becomes a product of $R$-matrices and the twisted translation 
operator:
\begin{align}
  T^{(2)}(w_j|w_1,\dots, w_N) = [q][q^2]&\check R_{j-1,j}\left({\frac{w_j}{w_{j-1}}}
\right)\cdots \check R_{1,2}\left({\frac{w_j}{w_{1}}}\right)S'
  \label{eqn:scatteringop}
  \\
  &\times \check R_{N-1,N}\left({\frac{w_j}{w_{N}}}\right)\cdots \check R_{j,j
+1}\left({\frac{w_j}{w_{j+1}}}\right) \nonumber.
\end{align}
The result is sometimes called a scattering operator as its effect is to drag the 
inhomogeneity at position $j$ through all others in a cyclic manner.
  
We equate both expressions, use the exchange relation 
\eqref{eqn:exchangeeigenvector}, and find after elimination of some trivial factors 
the following relation for all $j=1,\dots,N-1$:
\begin{align}
    \prod_{k=j+1}^N\left[{qw_j}/{w_k}\right] 
    &\prod_{1\leq k \leq j-1}^\curvearrowleft  \check R_{k,k+1}\left({{w_j}/{w_{k}}}
\right)S'|\Psi(w_1,\dots,w_{j-1},w_{j+1},\dots,w_N,w_j)\rangle
    \nonumber \\
    &= (-1)^{N+1}\prod_{k=1,\, k \neq j}^N\left[{qw_k}/{w_j}\right]\prod_{k=1}
^{j-1}\left[{q^2w_j}/{w_k}\right]|\Psi(w_1,\dots,w_N)\rangle.
    \label{eqn:scattering}
  \end{align}
The arrow $\curvearrowleft$ indicates that the product over the $R$-matrices is 
taken in reverse order. This equation states the equality of two vectors 
which are symmetric Laurent polynomials in each variable. In particular, they have 
the same zeroes, and the same degree for their leading 
and trailing terms. We conclude our eigenvector on the 
right-hand side is proportional to the first factor of the left-hand side
  \begin{equation*}
    \prod_{k=j+1}^N [qw_j/w_k]\quad \text{for all } j=1,\dots, N-1,
  \end{equation*}
  and hence to the product of all these factors. This proves the claim. \qed
\end{proof}

The preceding proposition suggests to divide out the redundant factors. Therefore, 
we introduce the renormalised vector
\begin{equation}
  |\tilde \Psi(w_1,\dots,w_N)\rangle =\left(([q][q^2])^{N/2}\prod_{1\leq j < k \leq 
N} [qw_j/w_k]\right)^{-1}|\Psi(w_1,\dots,w_N)\rangle.
  \label{eqn:defpsitilde}
\end{equation}
The additional division by $([q][q^2])^{N/2}$ is chosen for convenience in order to 
remove common multiplicative constants which come from the spin flips which the $
\mathcal B$-operator induces.

\paragraph{Recurrence.} The vector $|\tilde \Psi(w_1,\dots,w_N)\rangle$ obeys an 
exchange relation, too. Indeed, comparing with \eqref{eqn:exchangeeigenvector} we obtain 
for all $j=1,\dots,N-1$
\begin{subequations}
\begin{align}
  \check R_{j,j+1}\left(\frac{w_j}{w_{j+1}}\right)|\tilde \Psi(\dots, w_j,w_{j+1},
\dots)\rangle = \left[\frac{qw_{j+1}}{w_{j}}\right]\left[\frac{q^2w_j}{w_{j+1}}
\right]|\tilde \Psi(\dots, w_{j+1},w_{j},\dots)\rangle,
  \label{eqn:exchangephi}
\end{align}
whose right-hand side differs slightly from \eqref{eqn:exchangeeigenvector}.
Moreover, the renormalised vector transforms covariantly under cyclic shifts. 
Indeed, using \eqref{eqn:scattering} with $j=1$, we obtain after a slight 
redefinition of the variables the equation
\begin{equation}
  S'|\tilde \Psi(w_1,\dots,w_{N-1},w_N) \rangle = (-1)^{N+1} |\tilde 
\Psi(w_N,w_1,\dots,w_{N-1})\rangle.
  \label{eqn:cyclicphi}
\end{equation}
\label{eqn:qkzsystem}%
\end{subequations}
The equations \eqref{eqn:exchangephi} and \eqref{eqn:cyclicphi} are akin to the so-called quantum Knizhnik-Zamolodchikov system which was used in order to study 
similar problems. They allow to shuffle the inhomogeneity parameters around. 
Moreover, they allow to establish a relation between the vectors for lengths $N$ and 
$N-2$. This can be seen as follows: if we replace $w_j \to q^{-1}w_j$ and $w_{j+1} 
\to w_j$ then the exchange relation becomes
\begin{equation*}
  |\tilde \Psi(\dots,w_j,q^{-1}w_j,\dots)\rangle =([q][q^2])^{-1} \check R_{j,j+1}
(q^{-1})|\tilde \Psi(\dots,q^{-1}w_j,w_j,\dots)\rangle.
\end{equation*}
For $z=q^{-1}$ the $R$-matrix of the nineteen-vertex model is of rank one as was 
shown in \eqref{eqn:rankone}: we have $\check R(q^{-1}) = [q][q^2] |s\rangle\langle 
s|$ where $|s\rangle = |{\Uparrow\Downarrow}\rangle + |{\Downarrow\Uparrow}\rangle - 
|{00}\rangle$. We conclude that if $w_{j+1}=q^{-1}w_j$ then the vector is 
proportional to $|s\rangle$ on sites $j,j+1$. As we are free to permute the 
inhomogeneity parameters with the help of the exchange relation, it is sufficient to 
consider the case $j=1$. Taking advantage of the explicit form of the renormalised 
eigenvector in terms of $\mathcal B$-operators, one finds after a calculation the 
following relation:
\begin{equation*}
  |\tilde \Psi(w_1,q^{-1}w_1,w_3,\dots,w_N)\rangle = (-1)^N [q]\left(\prod_{j=3}^N 
\left[\frac{qw_1}{w_j}\right]\left[\frac{q^2w_j}{w_1}\right]\right) |s\rangle 
\otimes |\tilde \Psi(w_3,\dots,w_N)\rangle.
\end{equation*}
This equation relates the eigenstates at $N$ and $N-2$ through specialisation of a pair of inhomogeneity parameters.

\paragraph{Degree width.} Let us consider a simple example for the components of the 
renormalised vector $|\tilde \Psi\rangle$. Because of covariance under 
cyclic shifts we indicate the expressions only for representatives. For $N=3$ we 
find the expressions
\begin{align*}
  &\tilde \Psi_{\Uparrow 0 \Downarrow}(w_1,w_2,w_3)=\tilde \Psi_{\Downarrow 0 
\Uparrow}
  (w_1,w_2,w_3)=[q][qw_2/w_1][qw_3/w_2],\\
  & \tilde \Psi_{000}(w_1,w_2,w_3)=[q]^2[q^2]+[w_2/w_1][w_3/w_1][w_3/w_2].
\end{align*}
The components are thus centred Laurent polynomials in each variable $w_j$. The 
degree width of a Laurent polynomial is the difference of the degrees of its leading 
and trailing terms: for instance, $\sum_{j=-m}^n a_j w^j$ with $a_{-m},a_n \neq 0$ 
has degree width $m+n$ with respect to $w$. From our example for $N=3$, we see that 
the degree width of the eigenvector in a given variable varies from component to 
component. The degree width of the vector $|\tilde \Psi(w_1,\dots,w_N)\rangle$ in 
$w_j$ is the maximum degree width of its components.
\begin{proposition}
  The degree width of the vector $|\tilde \Psi(w_1,\dots,w_N)\rangle$ in each 
variable is $2(N-1)$. 
\label{prop:degreewidth}
\end{proposition}
\begin{proof}
  The degree width is obviously and additive quantity under multiplication of 
Laurent polynomials. Consider the operator $\mathcal B(z|w_1,\dots,w_N)$. It is a 
centred Laurent polynomial in all its arguments of degree width $2N$ in $z$, and $2$ 
in $w_k$ for all $k=1,\dots,N$. If we specify however $z=w_j$ for some $j$ then the 
monodromy matrix, and thus $\mathcal B(w_j|w_1,\dots,w_N)$, contains an $R$-matrix 
$R^{(1,2)}_{a,j}(q^{-1})$. Therefore the degree width of $\mathcal B(w_j|
w_1,\dots,w_N)$ is $2(N-1)$ in $w_j$, and $2$ in all other $w_k, \, k\neq j$. The 
construction \eqref{eqn:eigenvector} implies thus that $|\Psi(w_1,\dots,w_N)\rangle$ 
has degree width $4(N-1)$ in each variable because of the additivity property. In 
\eqref{eqn:defpsitilde} we divide out a multi-variable Laurent polynomial which is 
obviously of degree width $2(N-1)$ in each of its variables. Hence, in any $w_j$ the 
degree width of $|\tilde \Psi(w_1,\dots,w_N)\rangle$ is $4(N-1)-2(N-1) = 2(N-1)$. 
\qed
\end{proof}
It is possible to work out relations between components in the limits where $w_j \to 0$ or $w_j \to \infty$, i.e. the highest components of the leading and trailing parts of the vector as a Laurent polynomial in $w_j$. These follow straightforwardly from the graphical representation of the Bethe vector and lead to another recurrence relation. The key observation is to notice that the $R$-matrix $R^{(1,2)}(z)$ is diagonal at leading order when its argument is sent to zero or infinity. This allows to delete the $j$-th row and column from the picture \eqref{eqn:graphicalrepr}. After proper normalisation and some calculation, one finds the following result:
\begin{proposition} In the limit of infinite or zero spectral parameter $w_j$ the components of the renormalised Bethe vector for $N$ sites reduce at leading order to the components for $N-1$ sites according to
\begin{align*}
   \lim_{w_{j}\to\infty}w_j^{-(N-1)}&\tilde\Psi_{\cdots \sigma_{j-1}\sigma_j\sigma_{j+1}\cdots}(w_1,\dots,w_j,\dots,w_N)\\
   &=(-1)^{N-j}\delta_{\sigma_j,0}\left(\prod_{k\neq j}^N w_k^{-1}\right)\tilde\Psi_{\cdots \sigma_{j-1}\sigma_{j+1}\cdots}(w_1,\dots,w_{j-1},w_{j+1},\dots,w_N),
\end{align*}
and 
\begin{align*}
\lim_{w_{j}\to0}w_j^{N-1}&\tilde\Psi_{\cdots \sigma_{j-1}\sigma_j\sigma_{j+1}\cdots}(w_1,\dots,w_j,\dots,w_N)\\& =(-1)^{j-1}\delta_{\sigma_j,0}\left(\prod_{k\neq j}^N w_k\right)\tilde\Psi_{\cdots \sigma_{j-1}\sigma_{j+1}\cdots}(w_1,\dots,w_{j-1},w_{j+1},\dots,w_N),
\end{align*}
where $\delta_{ab} = 1$ for $a=b$, and $0$ otherwise.
  \label{prop:asymptotics}
\end{proposition}

\subsection{Scalar products and the square norm}
\label{sec:scalar}
The completely explicit nature of the Bethe roots \eqref{eqn:betheroots} implies 
that all components of the eigenvectors $|\Psi(w_1,\dots,w_N)\rangle$ and 
$|\tilde \Psi(w_1,\dots,w_N)\rangle$ can in principle be computed in finite size. It is desirable 
to develop a systematic practical scheme to do this for arbitrary finite $N$. We 
leave this to a future investigation. In this article we evaluate only a simple component
which turns out to be intimately related to the square norm of the vector. Hence, we discuss first the square norm in this section. 

A natural generalisation of the square norm of the eigenstates for the spin-one XXZ 
chain is the infinite-cylinder partition function (see \cite{zinnjustin:13} for an 
explanation of this interpretation) :
\begin{align*}
  Z(w_1,\dots,w_N) &= \langle \tilde \Psi(w_1^{-1},\dots,w_N^{-1})|\tilde 
\Psi(w_1,\dots,w_N)\rangle\\
&= \sum_{\sigma \in \{\Uparrow,0,\Downarrow\}^N} \tilde \Psi_{\sigma_1\cdots
\sigma_N}(w_1^{-1}, \dots, w_N^{-1})\tilde \Psi_{\sigma_1\cdots
\sigma_N}(w_1, \dots, w_N).
\end{align*}
As in the homogeneous case, we use the \textit{real} scalar product on $V$. In order to evaluate $Z(w_1,\dots,w_N)$ explicitly, we rewrite it in terms of the 
elements $\mathcal B,\,\mathcal C$ of the monodromy matrix, and use the theory of 
scalar products of the algebraic Bethe ansatz. To this end, we use the fact that 
under transposition the operator $\mathcal B$ can be transformed to $\mathcal C$ 
according to
\begin{equation*}
  \mathcal B(z^{-1}|w_1^{-1},\dots,w_N^{-1})^t = (-1)^{N-1}\mathcal C(z|
w_1,\dots,w_N).
\end{equation*}
This is a direct consequence of the so-called crossing symmetry of 
$R^{(1,2)}(z)$: the transpose with respect to the second space it acts on can be 
written as
\begin{equation*}
  R^{(1,2)}(z)^{t_2}=(\sigma^2\otimes 1)R^{(1,2)}(q^{-2}z^{-1})^{t_2}
(\sigma^2\otimes 1),\quad \sigma^2 =\left(
  \begin{array}{cc}
    0 & - \i\\
    \i & 0
  \end{array}
  \right).
  \end{equation*}
Combining this with the definition of the renormalised Bethe vector 
\eqref{eqn:defpsitilde} we obtain
\begin{equation}
  Z(w_1,\dots,w_N) =\frac{\langle \wedge|\prod_{j=1}^N \mathcal C(w_j|
w_1,\dots,w_N) \prod_{j=1}^N \mathcal B(w_j|w_1,\dots,w_N)|\wedge\rangle}
{(-1)^N[q^2]^N \prod_{j,k=1}^N [q^{-1}w_j/w_k]}.
  \label{eqn:zbc}
\end{equation}
Scalar products of this type can be computed with the help of determinant formulae thanks to
Slavnov's formula \cite{slavnov:89}. The original proof was given for periodic
boundary conditions but the generalisation to twisted boundary conditions is immediate :
\begin{theorem}[Slavnov's formula] Let $n>0$ and consider a solution $z_1,\dots, z_n
$ to the Bethe ansatz equations \eqref{eqn:bae} for $N$ sites and a set of arbitrary 
numbers $\zeta_1,\dots,\zeta_n$. Then the scalar product
\begin{equation*}
  S_n= \langle \wedge|\prod_{j=1}^n \mathcal C(z_j) \prod_{j=1}^n \mathcal 
B(\zeta_j)|\wedge\rangle
\end{equation*}
is given by
\begin{equation*}
  S_n = \prod_{j=1}^n d(z_j) d(\zeta_j)\prod_{1\leq k<j\leq n}g(z_j,z_k)g(\zeta_k,
\zeta_j)\prod_{j,k=1}^n \frac{f(z_j, \zeta_k)}{g(z_j, \zeta_k)}\,\det M,
\end{equation*}
where $M = (M_{jk})$ is an $n \times n$ matrix with entries
\begin{equation*}
  M_{jk} = e^{-\i \phi}\frac{g(z_j,\zeta_k)^2}{f(z_j,\zeta_k)}-
\frac{g(\zeta_k,z_j)^2}{f(\zeta_k,z_j)}\frac{a(\zeta_k)}{d(\zeta_k)}\prod_{m=1}^n 
\frac{f(\zeta_k,z_m)}{f(z_m,\zeta_k)}, \quad j,k=1,\dots,n.
\end{equation*}
\end{theorem}
For our problem the functions $a(z),d(z),g(z,w),f(z,w)$ were defined in section 
\ref{sec:aba}. The result is in fact very strong as it does not depend on their 
precise form but just a few analyticity properties.

The expression simplifies considerably if we specify for $\phi=\pi,\,n=N$ the Bethe 
roots to \eqref{eqn:betheroots}, and make use of the definition of $a(z), d(z)$. In 
this case, we obtain after some algebra the expression
\begin{equation}
  S_N = (-1)^N \left(\prod_{j=1}^N d(w_j)\right) Z_{\text{\tiny IK}}(\zeta_1,\dots,
\zeta_N;w_1,\dots,w_N).
  \label{eqn:scalarprod}
\end{equation}
Here $Z_{\text{\tiny IK}}(\zeta_1,\dots,\zeta_N;z_1,\dots,z_N)$ is given by the
so-called Izergin-Korepin determinant formula \cite{izergin:92}:
\begin{equation}
  Z_{\text{\tiny IK}}(\zeta_1,\dots,\zeta_N;w_1,\dots,w_N) = \frac{\prod_{j,k=1}^N 
\mathfrak a(\zeta_j/w_k)\mathfrak b(\zeta_j/w_k)}{\prod_{j<k}[\zeta_j/\zeta_k][w_k/
w_j]}\det_{j,k=1,\dots,N}\left(\frac{\mathfrak c(\zeta_j/w_k)}{\mathfrak a(\zeta_j/
w_k)\mathfrak b(\zeta_j/w_k)}\right).
  \label{eqn:IKdet}
\end{equation}
The functions $\mathfrak a(z), \mathfrak b(z),\mathfrak c(z)$ are 
are the statistical weights for the first, second, and third group of vertices of a 
six-vertex model, shown in figure \ref{fig:6v}. 
\begin{figure}[h]
  \centering
  \begin{tikzpicture}
    \drawvertex{1}{0}{0}
    \drawvertex{2}{1.25}{0}
    \drawvertex{3}{3.25}{0}
    \drawvertex{4}{4.5}{0}
    \drawvertex{5}{6.5}{0}
    \drawvertex{6}{7.75}{0}
    
    \draw (1.125,-1) node {$\mathfrak a(z) = [qz]$};
    \draw (4.425,-1) node {$\mathfrak b(z) = [q /z]$};
    \draw (7.425,-1) node {$\quad \mathfrak c(z) = [q^2]
$};
    \end{tikzpicture}
 \caption{Vertex configurations of the six-vertex model. A vertex of the first, 
second or third group has weight $\mathfrak a(z),\,\mathfrak b(z)$ or $\mathfrak 
c(z)$ respectively.}
 \label{fig:6v}
\end{figure}
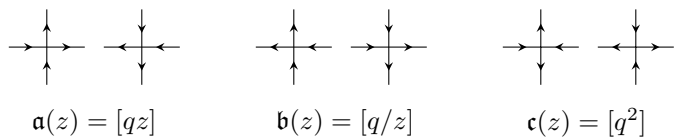
The Izergin-Korepin determinant \eqref{eqn:IKdet} is the partition function of the 
inhomogeneous six-vertex model on an $N\times N$ square lattice with so-called 
domain wall boundary conditions as illustrated on figure \ref{fig:dwbc}. At the top 
and bottom row of the square all arrows are outgoing, whereas they are ingoing at its left- and 
rightmost column. The weight of a vertex in row $j$ and column $k$ 
is chosen with spectral parameter $z= \zeta_j/w_k$ for $j,k=1,\dots,N$. The weight 
of a configuration is the product of the weights for each of its vertices.
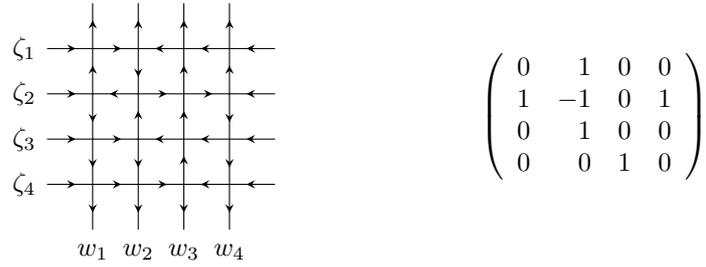
\begin{figure}
  \centering
  \begin{tikzpicture}[scale=1.2]

  \foreach \x in {0,.5,1,...,1.5}
  {
    \draw[postaction={on each segment={mid arrow}}] (\x,0) -- (\x,.5);
    \draw[postaction={on each segment={mid arrow}}] (-.5,-\x) -- (0,-\x);

    \draw[postaction={on each segment={mid arrow}}] (\x,-1.5) -- (\x,-2);
    \draw[postaction={on each segment={mid arrow}}] (2,-\x) -- (1.5,-\x);  
  }
  
  \draw[postaction={on each segment={mid arrow}}] (0,0) -- (.5,0) -- (.5,-.5) -- 
(1,-.5)--(1.5,-.5) -- (1.5,-1) -- (1.5,-1.5);
  \draw[postaction={on each segment={mid arrow}}] (1.5,-.5)--(1.5,0) -- (1,0) -- (.
5,0);
  \draw[postaction={on each segment={mid arrow}}] (1.5,-1.5)--(1,-1.5) -- (1,-1) -- 
(1,-.5) -- (1,0);
  \draw[postaction={on each segment={mid arrow}}] (1.5,-1)--(1,-1) -- (.5,-1) -- (.
5,-.5)--(0,-.5)--(0,0);
  \draw[postaction={on each segment={mid arrow}}] (0,-.5)--(0,-1) -- (.5,-1) -- (.
5,-1.5)--(1,-1.5);
  \draw[postaction={on each segment={mid arrow}}] (0,-1)--(0,-1.5) -- (.5,-1.5);
  
  \draw (0,-2.25) node {$w_1$}; 
  \draw (.5,-2.25) node {$w_2$}; 
  \draw (1,-2.25) node {$w_3$}; 
  \draw (1.5,-2.25) node {$w_4$};  
  
  \draw (-.75,0) node {$\zeta_1$}; 
  \draw (-.75,-.5) node {$\zeta_2$}; 
  \draw (-.75,-1) node {$\zeta_3$}; 
  \draw (-.75,-1.5) node {$\zeta_4$};  
   \draw (5.5,-.75) node 
    {
      $\left(
         \begin{array}{rrrr}
           0 & 1 & 0 & 0\\
           1 & -1 & 0 & 1\\
           0 & 1 & 0 & 0\\
           0 & 0 & 1 & 0
         \end{array}
       \right)
      $
    };

  \end{tikzpicture}
  \caption{\textit{Left:} $N\times N$ square lattice with domain-wall boundary 
conditions for $N=4$. The weight of a vertex in row $j$ and column $k$ is a function 
of the spectral parameters $\zeta_j$ and $w_k$. \textit{Right:} The corresponding 
alternating sign matrix.}
  \label{fig:dwbc}
\end{figure}

If we set $\zeta_j=w_j$ for all $j=1,\dots,N$, and combine the expression of the 
partition function \eqref{eqn:zbc} with our findings, then we obtain the following 
closed form:
\begin{proposition}
The renormalised Bethe vector $|\tilde \Psi(w_1,\dots,w_N)\rangle$ satisfies the 
following sum rule:
\begin{equation*} 
  Z(w_1,\dots,w_N) =[q^2]^{-N}Z_{\text{\rm \tiny IK}}(w_1,\dots,w_N;w_1,\dots,w_N).
\end{equation*}
\label{prop:partfunc}
\end{proposition}

Apart from leading to an explicit formula for the square norm of $|\tilde \Psi
\rangle$, the identification \eqref{eqn:scalarprod} has another interesting 
consequence. Notice that if we choose $\zeta_1,\dots,\zeta_N$ to be a solution of 
the Bethe equations with $\phi=\pi, n=N$ which leads to another eigenvalue, then the 
scalar product has to be zero. This implies that all such solutions need to solve 
$Z_{\text{\tiny IK}}(\zeta_1,\dots,\zeta_N;w_1,\dots,w_N)=0$ (in addition to the Bethe
equations): they lie in the 
algebraic variety defined by the zeroes of the Izergin-Korepin determinant in the 
first set of variables. It might be interesting to study the nature of these points 
on the variety.

\subsection{Simple components}
\label{sec:simplecomponents}
In this section we use the closed expressions for scalar products involving the vector $|\tilde \Psi(w_1,\dots,w_N)\rangle$ in order to determine the components
\begin{align*}
  \tilde \Psi_{ \underset{n}{\underbrace{\Uparrow \cdots \Uparrow}}\underset{n}{\underbrace{\Downarrow\cdots \Downarrow}}}(w_1,\dots,w_{2n}) \quad \text{and} \quad \tilde \Psi_{\underset{n}{\underbrace{\Uparrow \cdots \Uparrow}}0\underset{n}{\underbrace{\Downarrow\cdots \Downarrow}}}(w_1,\dots,w_{2n+1})
\end{align*}
for even length $N=2n$ and odd length $N=2n+1$ respectively. They are the most simple components in the sense that they can be evaluated as a product of monomials and the Izergin-Korepin determinant.

\paragraph{Even lengths.} We start with even $N=2n$. We claim that the component is given by the following matrix element for a system of length $n=N/2$:
\begin{align}
  \tilde \Psi_{\Uparrow\dots \Uparrow\Downarrow\cdots \Downarrow}(w_1,\dots, w_{2n})
    = \frac{\prod_{j=1}^{2n}\prod_{k=1}^n[q^{-1}w_j/w_k]}{([q][q^2])^n\prod_{j<k}^{2n}[q w_j/w_k]}\,\langle \vee | \prod_{j=1}^{2n} \mathcal B(w_j|w_{n+1},\dots,w_{2n})|\wedge\rangle.
    \label{eqn:componentfromB}
\end{align}
Here $|\vee\rangle = |{\Downarrow\cdots \Downarrow}\rangle$ is the spin-reversed Bethe reference state $|\wedge\rangle = |{\Uparrow\cdots \Uparrow}\rangle$ for $n$ sites. This surprising reduction of the system size is most easily seen from the graphical representation for the corresponding Bethe vector component $\Psi_{\Uparrow\dots \Uparrow\Downarrow\cdots \Downarrow}(w_1,\dots, w_{2n})$. Let us use the recipe of section \ref{sec:elementary}, and write it down graphically as the left-hand side of the following equation:
\begin{equation}
\begin{tikzpicture}[baseline=1.15cm]
\begin{scope}
\draw[rounded corners=3pt, color=lightgray, fill=lightgray] (-.35,-.4) 
rectangle(1.3,2.9);
  \foreach \x in {0,1,1.5,2.5}
{
  \draw[double,postaction ={on each segment={mid arrow}}] (\x,-.5) -- (\x,0);
\draw[double] (\x,0) -- (\x,2.5);
}
\draw (0,-.5) node[below] {\small $w_1$};
\draw (1,-.5) node[below] {\small $w_n$};
\draw (1.6,-.5) node[below] {\small $w_{n+1}$};
\draw (2.5,-.5) node[below] {\small $w_{2n}$};
\foreach \y in {0,1,1.5,2.5}
{
  \draw (0,\y) -- (2.5,\y);
  \draw[postaction ={on each segment={mid arrow}}]  (0,\y) -- (-.5,\y);
\draw[postaction ={on each segment={mid arrow}}]  (2.5,\y) -- (3,\y);
}
\foreach \x in {-.25,2.75}
{
  \foreach \y in {.625,2.125}
    \draw (\x,\y) node{$\vdots$};
}
\foreach \y in {-.25,2.75}
{
  \foreach \x in {.5,2.}
    \draw (\x,\y) node{$\dots$};
}
  \draw (-.5,0) node [left] {$w_1$};
  \draw (-.5,1) node [left] {$w_n$};
  \draw (-.5,1.5) node [left] {$w_{n+1}$};
  \draw (-.5,2.5) node [left] {$w_{2n}$};
\foreach \x in {0,1}
{
  \draw [double, postaction ={on each segment={mid arrow}}] (\x,2.5) -- (\x,3);
\draw [xshift=1.5cm, double, postaction ={on each segment={mid arrow}}] (\x,3) -- (\x,2.5);
}
\end{scope}
\begin{scope}[xshift=6.75cm]
\draw (-1.75,1.25) node {$=\prod_{j=1}^{2n}\prod_{k=1}^{n}\left[\frac{w_j}{q w_k}\right] \times$};
\foreach \x in {1.5,2.5}
{
  \draw[double,postaction ={on each segment={mid arrow}}] (\x,-.5) -- (\x,0);
\draw[double] (\x,0) -- (\x,2.5);
}
\draw (1.5,-.5) node[below] {\small $w_{n+1}$};
\draw (2.5,-.5) node[below] {\small $w_{2n}$};
\foreach \y in {0,1,1.5,2.5}
{
  \draw (1.5,\y) -- (2.5,\y);
  \draw[postaction ={on each segment={mid arrow}}]  (1.5,\y) -- (1,\y);
\draw[postaction ={on each segment={mid arrow}}]  (2.5,\y) -- (3,\y);
}
\foreach \x in {1.25,2.75}
{
  \foreach \y in {.625,2.125}
    \draw (\x,\y) node{$\vdots$};
}
\foreach \y in {-.25,2.75}
{
  \foreach \x in {2.}
    \draw (\x,\y) node{$\dots$};
}
  \draw (1,0) node [left] {$w_1$};
  \draw (1,1) node [left] {$w_n$};
  \draw (1,1.5) node [left] {$w_{n+1}$};
  \draw (1,2.5) node [left] {$w_{2n}$};
\foreach \x in {0,1}
{
\draw [xshift=1.5cm, double, postaction ={on each segment={mid arrow}}] (\x,3) -- (\x,2.5);
}
\end{scope}
\end{tikzpicture}
\label{eqn:sizereduction}
\end{equation}
In fact, the shaded region can be replaced by a simple product of Boltzmann weights. To see this, we inspect the vertex located at the upper-left corner of the picture: it has two outgoing arrows. The only possibility for it to have a non-vanishing Boltzmann weight is that its lower and right edge are ingoing. This gives the following configuration:
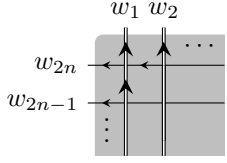
\begin{equation*}
  \begin{tikzpicture}
    \begin{scope}
    \clip (0,-.2) rectangle (1.8,1.5);
\draw[rounded corners=3pt, color=lightgray, fill=lightgray] (.1,1.4) rectangle (2.5,-1); 
    \draw[postaction={on each segment={mid arrow}}] (1,1) -- (.5,1) -- (0,1);
    \draw[postaction={on each segment={mid arrow}}] (.5,.5) -- (0,.5);
    \draw (1,1) -- (1.8,1);
    \draw (.5,.5) -- (1.8,.5);
    \draw[double] (.5,-.2) -- (.5,.5);
    \draw[double] (1,-.2) -- (1,1);
    \draw[double,postaction={on each segment={mid arrow}}] (.5,.5)--(.5,1) -- (.5,1.5);
    \draw[double] (1,0) -- (1,1);
    \draw[double,postaction={on each segment={mid arrow}}] (1,1) -- (1,1.5);
    \end{scope}
    
    \draw (0,1) node [left] {$w_{2n}$};
    \draw (0,.5) node [left] {$w_{2n-1}$};
    \draw (0.5,1.5) node [above] {$w_1$};
    \draw (1,1.5) node [above] {$w_2$};
    \draw (1.5,1.25) node {$\cdots$};
    \draw (0.25,0.25) node {$\vdots$};
  \end{tikzpicture}
\end{equation*}
We see that the vertices in the immediate neighbourhood of the corner have now two outgoing arrows, and hence their remaining edges need to be ingoing: this situation propagates, and allows to peel off row by row (or column by column) from the shaded region and replace the vertices by their weights $[q^{-1}w_j/w_k],\,j=1,\dots,2n,\, k=1,\dots,n$. This leads to the factor of the right-hand side in \eqref{eqn:sizereduction}. Using then the graphical representation of $\mathcal B(w|w_{n+1},\dots, w_{2n})$ as shown in \eqref{eqn:Bgraphical} leads, after correct normalisation, to the expression \eqref{eqn:componentfromB}.

The next step consists of converting \eqref{eqn:componentfromB} into a scalar product which can be computed with the help of Slavnov's formula. To this end, we need a short digression on the spin-reversal operator $\mathfrak R$. For length $N$, it acts like the matrix
\begin{equation*}
  \mathfrak R =
  \left(
    \begin{smallmatrix}
      0 & 0 & 1\\
      0 & 1 & 0\\
      1 & 0 & 0
    \end{smallmatrix}
  \right)^{\otimes N}.
\end{equation*}
It is not very difficult to show that the entries of the monodromy matrix can be related by spin reversal according to
\begin{equation*}
  \mathfrak R \mathcal A(z|w_1,\dots,w_N) = \mathcal D(z|w_1,\dots,w_N) \mathfrak R,\quad
  \mathfrak R \mathcal B(z|w_1,\dots,w_N) = \mathcal C(z|w_1,\dots,w_N) \mathfrak R
\end{equation*}
Hence the transfer matrices $T^{(j)}(w|w_1,\dots,w_N),\, j=1,2,$ with twist angle $\phi = \pi$ satisfy the (anti)com{\-}mutation relations
\begin{equation*}
  \mathfrak R T^{(j)}(z|w_1,\dots,w_N) = (-1)^j T^{(j)}(z|w_1,\dots,w_N) \mathfrak R.
\end{equation*}
Interestingly, the anticommutation relation between $T^{(1)}(z)$ and $\mathfrak R$ implies that for any eigenstate $|\Psi\rangle$ with $T^{(1)}(z)|\Psi\rangle = \theta^{(1)}(z)|\Psi\rangle$ the state $\mathfrak R|\Psi\rangle$ is also an eigenstate of $T^{(1)}(z)$ with eigenvalue $-\theta^{(1)}(z)$. Hence, all their simultaneous eigenvectors must be annihilated by $T^{(1)}(z)$. Unfortunately, this does not imply that a given eigenvector with $\theta^{(1)}(z)=0$, is also an eigenvector of $\mathfrak R$. To proceed, we make use of conjecture \ref{conj:uniqueness} which claims that the eigenspace of $\theta^{(1)}(z)=0$ is one-dimensional for any $N$ (except for possibly a finite number of values for $q$). Therefore, our (renormalised) Bethe state needs to be an eigenvector of the spin-reversal operator for any $N$. The recurrence relations of proposition \ref{prop:asymptotics} imply that the eigenvalue needs to be the same for all $N$. For $N=1$ we have simply $|\tilde \Psi(w_1)\rangle = |0\rangle$ which is invariant under spin reversal: $\mathfrak R|\tilde \Psi(w_1)\rangle = |\tilde \Psi(w_1)\rangle$. Hence, for arbitrary $N$ we obtain the relation
\begin{equation*}
  \mathfrak R|\tilde \Psi(w_1,\dots,w_N)\rangle = |\tilde \Psi(w_1,\dots,w_N)\rangle,
\end{equation*}
and of course the same equation for $|\Psi(w_1,\dots,w_N)\rangle $. Written out explicitly in terms of the operators $\mathcal B$ and $\mathcal C$ we find therefore
\begin{equation}
  \prod_{j=1}^N \mathcal C(w_j|w_1,\dots,w_N)|\vee\rangle = \prod_{j=1}^N \mathcal B(w_j|w_1,\dots,w_N)|\wedge\rangle.
  \label{eqn:spinreversal}
\end{equation}

Let us now come back to the evaluation of our component through \eqref{eqn:componentfromB}. We use the transposed version of \eqref{eqn:spinreversal} for $n$ sites in order to convert half of the $\mathcal B$-operators into $\mathcal C$'s:
\begin{align*}
  \tilde \Psi_{\Uparrow\dots \Uparrow\Downarrow\cdots \Downarrow}(w_1,\dots, w_{2n})
    =& \frac{\prod_{j=1}^{2n}\prod_{k=1}^n[q^{-1}w_j/w_k]}{([q][q^2])^n\prod_{j<k}^{2n}[q w_j/w_k]}\\
    & \times \langle \wedge |\prod_{j=n+1}^{2n} \mathcal C(w_j|w_{n+1},\dots,w_{2n}) \prod_{j=1}^{n} \mathcal B(w_j|w_{n+1},\dots,w_{2n})|\wedge\rangle.
\end{align*}
We see thus appear a typical scalar product which can be evaluated from \eqref{eqn:scalarprod} for a system of length $n$. Hence, the component can be written in terms of the Izergin-Korepin determinant. Some of the pre-factors cancel out and the final result takes the compact form
\begin{align}
  \tilde \Psi_{\Uparrow\dots \Uparrow\Downarrow\cdots \Downarrow}(w_1,\dots, w_{2n}) =
  \left(\frac{[q]}{[q^2]}\right)^n  \prod_{1\leq j < k \leq n}& \left[\frac{q w_k}{w_j}\right]\prod_{n+1\leq j < k \leq 2n} \left[\frac{q w_k}{w_j}\right] \nonumber \\
    &\times Z_{\text{\rm \tiny IK}}(w_1,\dots,w_n;w_{n+1},\dots, w_{2n}).
    \label{eqn:resulteven}
\end{align}

\paragraph{Odd lengths.} For $N=2n+1$ the component $\tilde \Psi_{\Uparrow \cdots \Uparrow 0 \Downarrow\cdots \Downarrow}(w_1,\dots,w_{2n+1})$ can be obtained from the result at even length, the known degree width of the vectors and the exchange relation \eqref{eqn:exchangephi}. The latter can be projected on spin configurations which allows us to investigate the dependence of the components on their parameters. For instance, it is easily shown that
\begin{align*}
  \left[\frac{q w_{j+1}}{w_{j}}\right] \tilde \Psi_{\cdots \underset{j}{\Uparrow} \Uparrow \cdots}(\dots,w_j,w_{j+1},\dots) =\left[\frac{q w_j}{w_{j+1}}\right] \tilde \Psi_{\cdots \underset{j}{\Uparrow} \Uparrow \cdots}(\dots,w_{j+1},w_j,\dots).\end{align*}
This implies in particular that any component with the pattern $\Uparrow \Uparrow$ at positions $j,j+1$ is proportional to $[qw_j/w_{j+1}]$ times a centred Laurent polynomial which is symmetric under the exchange of $w_j$ and $w_{j+1}$. Another simple consequence of the exchange relation is
\begin{align*}
   \tilde \Psi_{\cdots \underset{j}{\Uparrow} 0 \cdots}(\dots,w_j,w_{j+1},\dots) =\left[\frac{q w_{j+1}}{w_{j}}\right]\tilde \Upsilon_{\cdots \underset{j}{\Uparrow} 0 \cdots}(\dots, w_j,w_{j+1},\dots),
\end{align*}
where $\tilde \Upsilon_{\cdots {\Uparrow} 0 \cdots}(\dots, w_j,w_{j+1},\dots)$ is some centred Laurent polynomial in its arguments. Combining this with the symmetry property of the arguments within a string $\Uparrow \cdots \Uparrow$ allows to conclude that the component $\tilde \Psi_{\Uparrow \cdots \Uparrow 0 \Downarrow\cdots \Downarrow}(w_1,\dots,w_{2n+1})$ is actually proportional to $[q w_{n+1}/w_j]$ for all $j=1, \dots, n$, and hence to their product. A similar relation is found when analysing the dependence on the parameters to the right of the site $n+1$. Factor exhaustion leads thus to the following form for our component:
\begin{align*}
   \tilde \Psi_{\Uparrow\dots \Uparrow 0\Downarrow\cdots \Downarrow}(w_1,\dots, w_{2n}) =
   \prod_{j=1}^n & \left[\frac{q w_{n+1}}{w_j}\right]\prod_{j=n+2}^{2n+1}\left[\frac{q w_{j}}{w_{n+1}}\right] \nonumber \\
   &\times \tilde \Xi(w_1,\dots,w_n,w_{n+1},w_{n+2},\dots, w_{2n+1}).
\end{align*}
Here $\tilde \Xi(w_1,\dots,w_{2n+1})$ is a centred Laurent polynomial which can be easily determined from the recurrence properties of the renormalised eigenvector. Indeed, observe that the prefactor of $\tilde \Xi(w_1,\dots,w_{2n+1})$ in this equation is a centred Laurent polynomial in $w_{n+1}$ of degree width $4n = 2(N-1)$, and hence saturates the degree width according to proposition \ref{prop:degreewidth}. Therefore $\tilde \Xi(w_1,\dots,w_{2n+1})$ cannot depend on $w_{n+1}$, and may be evaluated by sending $w_{n+1}$ to infinity or zero. In these limits, we know the left-hand side from proposition \ref{prop:asymptotics}, and conclude that the unknown function $\tilde \Xi$ is given by the simple component \eqref{eqn:resulteven} for even size $N-1 =2n$ with slightly rearranged arguments. After a short calculation we find our final result for $N=2n+1$ :
\begin{align}
   \tilde \Psi_{\Uparrow\dots \Uparrow 0\Downarrow\cdots \Downarrow}(w_1,\dots, w_{2n+1}) =
   \prod_{j=1}^n & \left[\frac{q w_{n+1}}{w_j}\right]\prod_{j=n+2}^{2n+1}\left[\frac{q w_{j}}{w_{n+1}}\right] \nonumber \\
   &\times \tilde \Psi_{\Uparrow\dots \Uparrow \Downarrow\cdots \Downarrow}(w_1,\dots,w_n,w_{n+2},\dots, w_{2n+1}).
   \label{eqn:resultodd}
\end{align}

\subsection{The homogeneous point}
\label{sec:homogeneous}
In this section, we consider the homogeneous point $w_1=\dots=w_N=1$. In this case, 
the vector $|\tilde\Psi(1,\dots,1)\rangle$ is an eigenstate of the spin-chain 
Hamiltonian with eigenvalue zero in the pseudo-momentum sector where $S'\equiv (-1)^{N+1}$, and thus a supersymmetry singlet. Our aim is to 
obtain the expressions for the simple components \eqref{eqn:simplecomponents} and the sum rule \eqref{eqn:sumrule}. To this end, we need to 
control the normalisation of the vector in the homogeneous case. We show here that 
there is an redundant overall factor for all components when all $w$'s are equal. The case of $N=3$ sites provides a simple illustration :
\begin{align*}
  &\tilde \Psi_{\Uparrow 0 \Downarrow}(1,1,1)=\tilde \Psi_{\Downarrow 0 \Uparrow}
  (1,1,1)=[q]^3 \times 1, \quad \tilde \Psi_{000}(1,1,1)=[q]^3\times x, \quad x =q
+q^{-1}.
\end{align*}
If we remove the factor $[q]^3$ then we recover the components of the eigenvector of 
the Hamiltonian $|\Phi(x)\rangle$ for three sites from section 
\ref{sec:defspinchain}. This can be done systematically for all $N$ as we show 
hereafter. After this, we discuss the results obtained in sections \ref{sec:scalar} and \ref{sec:simplecomponents}, and relate the supersymmetry singlet and enumeration problems of alternating sign matrices.

\paragraph{Limit of the renormalised Bethe vector.}  From its definition we find that 
the homogeneous limit of the renormalised Bethe vector is given by
  \begin{equation}
    |\tilde \Psi(1,\dots,1)\rangle =[q]^{N(N-1)/2}{x^{-N/2}}\beta(x)^N |{\wedge}
\rangle, \quad x=q+q^{-1},
    \label{eqn:hompsitilde}
  \end{equation}
  where $\beta(x) = \mathcal B(1|1,\dots,1)/[q]^N$. More explicitly, we have
  \begin{subequations}
  \begin{equation}
    \beta(x) = {}_{a}\langle{\uparrow}|\rho_{a,N}(x) \cdots \rho_{a,1}(x)|{\downarrow}\rangle_a, \quad 
\rho_{a,j}(x) = [q]^{-1}R^{(1,2)}_{a,j}(q^{-1}).
  \end{equation}
  In order to show that this depends only on $x$ it is sufficient to use the 
definition \eqref{eqn:r12}: we find 
  \begin{equation}
    \rho(x)   =\left(
  \begin{array}{cccccc}
  1 & 0 & 0 & 0 & 0 & 0\\
  0 & 0 & 0 & x^{1/2} & 0 & 0\\
  0 & 0 & -1 & 0 & x^{1/2} & 0\\
    0 & x^{1/2} & 0 & -1 & 0 & 0\\
  0 & 0 & x^{1/2} & 0 & 0 & 0 \\
  0 & 0 & 0 & 0 & 0 & 1
  \end{array}
  \right).
  \end{equation}
  \label{eqn:constructionbeta}%
  \end{subequations}
It is easily seen that when acting on a state $|\epsilon\rangle \otimes |\sigma
\rangle$ with $\epsilon = \uparrow,\downarrow$ and $\sigma=\Uparrow,0,\Downarrow$ 
that a spin flip of any type is weighted by $x^{1/2}$. For example $\rho(x)\left(|
{\uparrow}\rangle \otimes |{\Downarrow}\rangle\right) = -|{\uparrow}\rangle \otimes 
|{\Downarrow}\rangle+ x^{1/2}|{\downarrow}\rangle\otimes |0\rangle$.
The way the operator $\beta(x)$ flips spins is rather non-local. Yet it is quite straightforward to determine certain properties of its action 
on simple basis vectors of $V$:
\begin{lemma}
  Let $|\sigma_1,\dots, \sigma_N\rangle$ be a basis vector of $V$, then the state $
\beta(x)|\sigma_1,\dots, \sigma_N\rangle$ is $x^{1/2}$ times a polynomial in $x$ 
with integer coefficients.
   \label{lemma:actionbeta}
 \end{lemma}
 \begin{proof}
   The action of the operator $\beta(x)$ onto the basis state leads to a 
superposition of states where the spin components are flipped. From its construction 
\eqref{eqn:constructionbeta}, it is evident that any spin flip is weighted by 
$x^{1/2}$. A given spin component $\sigma_j$ may either be lowered if $\sigma_j=
\Uparrow,0$ or raised in case of $\sigma_j=0,\Downarrow$. Let us now consider those 
resulting vectors where exactly $k=0,1,\dots$ components are raised. This comes with 
a weight $x^{k/2}$. Moreover, their positions divide the chain into $k+1$ 
subsegments of length at least one whose horizontal ends have exactly the same 
boundary conditions as the operator $\mathcal B$ or $\beta$. This can easily be seen  
graphically by employing \eqref{eqn:Bgraphical}: an example with $k=2$ is shown in 
figure \ref{fig:tworeverseflips}.
   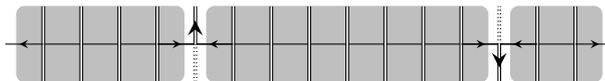
\begin{figure}[h]
     \centering
     \begin{tikzpicture}
       \draw[rounded corners=3pt, color=lightgray, fill=lightgray] (-3.85,-.5) 
rectangle(-1.65,.5);
       \draw[rounded corners=3pt, color=lightgray, fill=lightgray] (-1.35,-.5) 
rectangle(2.35,.5);
       \draw[rounded corners=3pt, color=lightgray, fill=lightgray] (2.65,-.5) 
rectangle(3.85,.5);

       \draw (-3.5,0) -- (3.5,0);
       \foreach \x in {-3.5,-3,...,-2,-1,-.5,...,2,3,3.5}
         \draw[double,fill=lightgray] (\x,-.5) -- (\x,.5);
         \draw[postaction={on each segment={mid arrow}}](-3.5,0) --(-4,0);
         \draw[postaction={on each segment={mid arrow}}](-2.,0) --(-1.5,0);
         \draw[postaction={on each segment={mid arrow}}](-1,0) --(-1.5,0);
         \draw[postaction={on each segment={mid arrow}}](2.,0) --(2.5,0);
         \draw[postaction={on each segment={mid arrow}}](3,0) --(2.5,0);
         \draw[postaction={on each segment={mid arrow}}](3.5,0) --(4,0);

       \draw[densely dotted,double] (-1.5,-.5)--(-1.5,0);
       \draw[postaction={on each segment={mid arrow}},double] (-1.5,0)--(-1.5,.5);  
       
       \draw[densely dotted,double] (2.5,0)--(2.5,.5);
       \draw[postaction={on each segment={mid arrow}},double] (2.5,0)--(2.5,-.5);

     \end{tikzpicture}
     \caption{Example for $k=2$ raising spin flips and the subdivision into $k+1 =3$ 
segments.}
     \label{fig:tworeverseflips}
   \end{figure}   
   As $\beta(x)$ decreases by construction the magnetisation of the state by one, we 
conclude that one these subsegments exactly $k+1$ spin components need to lowered. 
These spin flips come with a weight $x^{(k+1)/2}$. Therefore we obtain the total 
weight $x^{k/2}\times x^{(k+1)/2} = x^{1/2} \times x^{k}$ where $k$ is an integer. 
Taking into account all possible configurations, we conclude that the overall result 
is necessarily $x^{1/2}$ times a polynomial. The statement that it has integer 
coefficients follows from the fact that $\rho(x)$ is a linear function of $x^{1/2}$ 
whose coefficients are matrices with integer entries. \qed 
 \end{proof}
 
Next, we apply the operator $\beta(x)$ exactly $N$ times to the reference state. 
Using the preceding lemma we obtain immediately the following:
 \begin{proposition}
   The supersymmetry singlet
   \begin{equation}
     |\Phi(x)\rangle = x^{-N/2}\beta(x)^N|{\wedge}\rangle = [q]^{-N(N-1)/2}|\tilde \Psi(1,\dots,1)\rangle
     \label{eqn:defphi}
   \end{equation}
   is a polynomial in $x=q+q^{-1}$ with integer coefficients.
   \label{prop:poly}
\end{proposition}   

 \paragraph{Simple components and square norm.} We would 
like to evaluate the square norm $Z(w_1,\dots,w_N)$ and the components $\tilde \Psi_{\Uparrow \cdots \Uparrow\Downarrow\cdots \Downarrow}(w_1,\dots,w_{2n})$ or $\tilde \Psi_{\Uparrow \cdots \Uparrow0\Downarrow\cdots \Downarrow}(w_1,\dots,w_{2n+1})$ in the case $w_1=
\dots=w_N=1$. From our results in the preceding sections we see that this amounts to 
analyse the partition function of the six-vertex model with domain wall boundary 
conditions in the homogeneous limit. Its relation to weighted enumeration of alternating sign matrices is well known (see e.g. \cite{behrend:12}).

The relation is established as follows. In the homogeneous case the vertex weights become $\mathfrak a(1) 
= \mathfrak b(1) = [q]$ for the first four vertex configurations, and $\mathfrak 
c(1) = [q^2]$ for the fifth and sixth vertex configuration shown in figure \ref{fig:vertexentries} here below.
 Let us evaluate the weight of a single configuration on an $N\times N$ square with 
say $k$ vertices of type six.
It is not difficult to see that the boundary conditions imply that there are thus $N
+k$ vertices of type five, and
$N^2 - N - 2k$ vertices of the other types. We conclude that the weight of the 
configuration is therefore
\begin{equation*}
   [q]^{N^2 - N - 2k}[q^2]^{N+k} [q^2]^{k} = [q]^{N(N-1)}[q^2]^{N}(x^2)^k, \quad x 
=q+q^{-1}.
\end{equation*}
It is known that any such configuration is in one-to-one correspondence with an $N 
\times N$ alternating sign matrix containing exactly $k$ entries $-1$ \cite{elkies:92}. The correspondence between vertex configurations and matrix entries is shown in 
figure \ref{fig:vertexentries}, and we give an example in figure \ref{fig:dwbc}.
\begin{figure}[h]
  \centering
  \begin{tikzpicture}
    \drawvertex{1}{0}{0}
    \drawvertex{2}{1.25}{0}
    \drawvertex{3}{3.25}{0}
    \drawvertex{4}{4.5}{0}
    \drawvertex{5}{6.5}{0}
    \drawvertex{6}{7.75}{0}
    
    \draw (.5,-1) node {$0$};
    \draw (1.75,-1) node {$0$};
    \draw (3.75,-1) node {$0$};
    \draw (5,-1) node {$0$};
     \draw (7,-1) node {$1$};
    \draw (8.25,-1) node {$-1$};
 \end{tikzpicture}
 \caption{Correspondence between vertex configurations and the entries of alternating sign matrices.}
 \label{fig:vertexentries}
\end{figure}
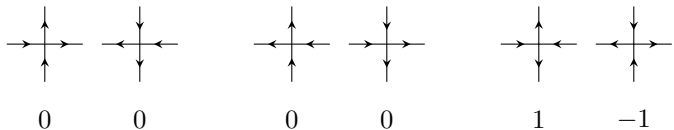
Therefore, in the homogeneous limit the Izergin-Korepin determinant reduces to a sum 
over all $N\times N$ alternating sign matrices, each of which is weighted by the 
weight $[q]^{N(N-1)}[q^2]^{N}(x^2)^k$ where $k$ is the number of minus signs. This 
leads to the type of weighted enumeration which we discussed in section 
\ref{sec:defspinchain}. In terms of the generating function $A_N(t)$ one finds therefore the partition function
\begin{equation*}
  Z_{\text{\rm \tiny IK}}(w_1=1,\dots,w_N=1) = [q]^{N(N-1)}[q^2]^{N} A_N(t=x^2).
\end{equation*}

This relation can now be applied to the vector $|\Phi(x)\rangle$ defined in \eqref{eqn:defphi}: sending all inhomogeneity parameters to one in \eqref{eqn:resulteven} and \eqref{eqn:resultodd} we find
\begin{equation*}
    \Phi_{\underset{n}{\underbrace{\Uparrow\cdots\Uparrow}}\underset{n}{\underbrace{\Downarrow\cdots\Downarrow}}}(x) = \Phi_{\underset{n}{ \underbrace{\Uparrow\cdots\Uparrow}}0\underset{n}{\underbrace{\Downarrow\cdots\Downarrow}}}(x) = A_n(x^2),
  \end{equation*}
which is precisely \eqref{eqn:simplecomponents}. In fact, the constant term of these polynomials equals $n!$, because any alternating sign matrix without minus signs is necessarily a permutation matrix. Moreover, the degree of the polynomial is $\lfloor (n-1)^2/4\rfloor$. Hence our normalisation criteria for the components from section \ref{sec:defspinchain} are fulfilled, and the state $|\Phi(x)\rangle$ is indeed a polynomial in $x$ with integer coefficients as established in proposition \ref{prop:poly}. 
Eventually, taking the homogeneous limit in proposition \ref{prop:partfunc}, we obtain
\begin{equation*}
  ||\Phi(x)||^2 = A_N(x^2)
\end{equation*}
which is the sum rule stated in \eqref{eqn:sumrule}. 
\section{Conclusion}
\label{sec:concl}
In this article we showed that for a particular twisted boundary condition the 
transfer matrix of the inhomogeneous nineteen-vertex model possesses a simple 
eigenvalue with a quite non-trivial corresponding eigenvector. At the homogeneous 
point, it is a supersymmetry singlet of the spin-one XXZ Hamiltonian, and its 
square norm as well as a particular component are given by generating functions for a type of weighted enumeration 
of alternating sign matrices, the weight being related to the anisotropy of the spin 
chain.

There are various open problems and generalisations. First of all, it might be interesting to characterise explicitly all components of the 
supersymmetry singlet in terms of combinatorial quantities, as was done for the spin-$1/2$ XXZ Hamiltonian with $
\Delta=-1/2$ where a relation to certain refined enumerations of alternating sign 
matrices was found \cite{razumov:06,razumov:07}. For instance, in \cite{hagendorf:13} certain components of the singlet where conjectured to be given by generating 
functions for weighted enumeration of vertically-symmetric alternating sign 
matrices. Second, one 
may ask if there are other boundary conditions/twists which allow for simple 
eigenvalues of the transfer matrix. The 
answer is affirmative: in fact, one may study systematically the boundary conditions 
which are compatible with dynamic lattice supersymmetry and admit supersymmetry 
singlets \cite{hagendorf:14susybc}, and use them in order to determine boundary 
conditions for which the transfer matrices of the corresponding inhomogeneous vertex 
model are expected to have a simple eigenvalue. Indeed, the nineteen-vertex model 
admits a non-diagonal twist which has all these features. This will be 
discussed in a forthcoming paper \cite{hagendorf:14spinreversal}. Third, the lattice 
supersymmetry of the spin-one XXZ chain survives the deformation to the quantum 
integrable spin-one XYZ chain \cite{hagendorf:13}. We expect the transfer matrix of 
the corresponding vertex model to have a simple eigenvalue 
but very non-trivial eigenvector, too, similarly to the case of the eight-vertex model along a 
particular line of couplings \cite{zinnjustin:13}. Eventually, it would be 
interesting to 
understand the role which dynamic lattice supersymmetry plays in all this. While 
this type of supersymmetry is present in quite a few integrable spin chains, its interplay 
with integrability is not very well understood (see for instance the discussion in 
\cite{meidinger:13}), not to mention the question if there 
is a general underlying structure in the case of inhomogeneous models which reduces 
to the lattice supersymmetry when approaching the homogeneous point.

\subsection*{Acknowledgements} This work is supported by the Belgian Interuniversity 
Attraction Poles Program P7/18 through the network DYGEST (Dynamics, Geometry and 
Statistical Physics). The author would like to thank Luigi Cantini, Anita Ponsaing, Paul Zinn-Justin for discussions.


\begin{thebibliography}{44}

\bibitem{baxterbook}
R.J. Baxter,
\newblock {\em {Exactly solved models in statistical mechanics}},
\newblock London Academic (1982).

\bibitem{gaudin:83}
M. Gaudin,
\newblock {\em La fonction d'onde de Bethe},
\newblock Masson Paris (1983).

\bibitem{fridkin:00}
V.~{Fridkin}, Y.~{Stroganov}  and D.~{Zagier},
\newblock {\em {Finite Size XXZ Spin Chain with Anisotropy Parameter $\Delta =
  1/2$}},
\newblock J. Stat. Phys. \textbf{102} (2000) 781--794.

\bibitem{fridkin:00_2}
V.~{Fridkin}, Y.~{Stroganov}  and D.~{Zagier},
\newblock {\em {Ground state of the quantum symmetric finite-size XXZ spin
  chain with anisotropy parameter $\Delta = {\frac12}$}},
\newblock J. Phys. A: Math. Gen. {\textbf{33}} (2000)   L121--L125,

\bibitem{bazhanov:05}
V.~V. {Bazhanov} and V.~V. {Mangazeev},
\newblock {\em {Eight-vertex model and non-stationary Lam{\'e} equation}},
\newblock J. Phys. A: Math. Gen. {\textbf{38}} (2005)   L145--L153.

\bibitem{bazhanov:06}
V.~V. {Bazhanov} and V.~V. {Mangazeev},
\newblock {\em {The eight-vertex model and Painlev{\'e} VI}},
\newblock J. Phys. A: Math. Gen. {\textbf{39}} (2006)   12235--12243,

\bibitem{zamolodchikov:81}
A.~B. Zamolodchikov and V.~A. Fateev,
\newblock {\em {A model factorized $S$-matrix and an integrable spin-$1$
  Heisenberg chain}},
\newblock Sov. J. Nucl. Phys. {\textbf{32}} (1981)   298--303.

\bibitem{fateev:81}
V.~A. Fateev,
\newblock {\em {A factorized $S$-matrix for particles of opposite parities and
  an integrable $21$-vertex statistical model}},
\newblock Sov. J. Nucl. Phys. {\textbf{33}} (1981)   761--766.

\bibitem{hagendorf:13}
C. Hagendorf,
\newblock {\em {Spin chains with dynamical lattice supersymmetry}},
\newblock J. Stat. Phys. {\textbf{150}} (2013)   609--657.

\bibitem{baxter:02}
R.~J. Baxter,
\newblock {\em {Completeness of the Bethe Ansatz for the Six and Eight-Vertex
  Models}},
\newblock J. Stat. Phys. {\textbf{108}} (2002)   1--48.

\bibitem{razumov:00}
A.~V. {Razumov} and Y.~G. {Stroganov},
\newblock {\em {Spin chains and combinatorics}},
\newblock J. Phys. A : Math. Gen. {\textbf{34}} (2001)   3185--3190,

\bibitem{razumov:01}
A.~V. {Razumov} and Y.~G. {Stroganov},
\newblock {\em {Spin chains and combinatorics: twisted boundary conditions}},
\newblock J. Phys. A: Math. Gen. {\textbf{34}} (2001)   5335--5340,

\bibitem{batchelor:01}
M.~T. {Batchelor}, J.~{de Gier}  and B.~{Nienhuis},
\newblock {\em {The quantum symmetric XXZ chain at $\Delta = - 1/2$ ,
  alternating-sign matrices and plane partitions}},
\newblock J. Phys. A: Math. Gen. {\textbf{34}} (2001)   L265--L270.

\bibitem{razumov:10}
A.~V. {Razumov} and Y.~G. {Stroganov},
\newblock {\em {A possible combinatorial point for the XYZ spin chain}},
\newblock Theor. Math. Phys. {\textbf{164}} (2010)   977--991,

\bibitem{mangazeev:10}
V.~V. {Mangazeev} and V.~V. {Bazhanov},
\newblock {\em {The eight-vertex model and Painlev{\'e} VI equation II:
  eigenvector results}},
\newblock J. Phys. A: Math. Theor. {\textbf{43}} (2010)   085206,

\bibitem{difrancesco:05_3}
P.~{Di Francesco} and P.~{Zinn-Justin},
\newblock {\em {Around the Razumov-Stroganov conjecture: proof of a
  multi-parameter sum rule}},
\newblock Electr. J. Comb. {\textbf{12}} (2005)  ~R6,

\bibitem{difrancesco:06}
P.~{Di Francesco}, P.~{Zinn-Justin}  and {J.-B.} {Zuber},
\newblock {\em {Sum rules for the ground states of the O(1) loop model on a
  cylinder and the XXZ spin chain}},
\newblock J. Stat. Mech. {\textbf{8}} (2006)  ~11,

\bibitem{zinnjustin:13}
P.~Zinn-Justin,
\newblock {\em {Sum Rule for the Eight-Vertex Model on Its Combinatorial
  Line}},
\newblock in Kenji Iohara, Sophie Morier-Genoud and Bertrand R{\'e}my
  (eds.), {\em Symmetries, Integrable Systems and Representations}, {\em {\em
  Vol,}~\bf 40}, 599--637, Springer-Verlag London (2013).

\bibitem{cantini:12_1}
L.~{Cantini},
\newblock {\em {Finite size emptiness formation probability of the XXZ spin
  chain at $\Delta=-1/2$}},
\newblock J. Phys. A: Math. Theor. {\textbf{45}} (2012)   135207.

\bibitem{korepin:93}
V.E. Korepin, N.M. Bogoliubov  and A.G. Izergin,
\newblock {\em {Quantum Inverse Scattering Method and Correlation Functions.}},
\newblock Cambridge University Press, 1993.

\bibitem{slavnov:89}
N.A. Slavnov,
\newblock {\em {Calculation of scalar products of wave functions and form
  factors in the framework of the algebraic Bethe ansatz}},
\newblock Theor. Math. Phys. {\textbf{79}} (1989)   502.

\bibitem{babujian:82}
H.~M. {Babujian},
\newblock {\em {Exact solution of the one-dimensional isotropic Heisenberg
  chain with arbitrary spins S}},
\newblock Phys. Lett. A {\textbf{90}} (1982)   479--482.

\bibitem{babujian:83}
H.~M. Babujian,
\newblock {\em Exact solution of the isotropic heisenberg chain with arbitrary
  spins: Thermodynamics of the model},
\newblock Nucl. Phys. B {\textbf{215}} {\textbf{3}} (1983)   317--336.

\bibitem{bares:90}
P.~A. {Bares} and G.~{Blatter},
\newblock {\em {Supersymmetric t-J model in one dimension: Separation of spin
  and charge}},
\newblock Phys. Rev. Lett. {\textbf{64}} (1990)   2567--2570.

\bibitem{meidinger:13}
D.~{Meidinger} and V.~{Mitev},
\newblock {\em Dynamic Lattice Supersymmetry in $\mathfrak{gl}(n|m)$ Spin Chains},
\newblock arXiv:1312.7021 (2013).

\bibitem{frahm:11_2}
{Frahm}, H., {Martins}, M.J.: {Finite size properties of staggered
  U$_{q}[sl(2|1)]$ superspin chains}.
\newblock Nucl. Phys. B \textbf{847}, 220--246 (2011).

\bibitem{zeilberger:96}
D. Zeilberger,
\newblock {\em Proof of the alternating sign matrix conjecture},
\newblock Electron. J. Comb. {\textbf{3}} (1996)   R13.

\bibitem{bressoudbook}
D. Bressoud,
\newblock {\em {Proofs and confirmations: the story of the alternating sign
  matrix conjecture}},
\newblock Cambridge University Press, 1999.

\bibitem{kuperberg:96}
G. Kuperberg,
\newblock {\em {Another proof of the alternating sign matrix conjecture}},
\newblock Internat. Math. Res. Notices {\textbf{3}} (1996)   139--150.

\bibitem{kuperberg:02}
G. Kuperberg,
\newblock {\em {Symmetry Classes of Alternating-Sign Matrices under One Roof}},
\newblock Ann. Math. {\textbf{156}} (2002)   835--866.

\bibitem{robbins:00}
D.~P. {Robbins},
\newblock {\em Symmetry Classes of Alternating Sign Matrices},
\newblock arXiv:math.CO/0008045 (2000).

\bibitem{behrend:12}
R.~E. Behrend, P.~Di Francesco  and P. Zinn-Justin,
\newblock {\em On the weighted enumeration of alternating sign matrices and
  descending plane partitions},
\newblock J. Comb. Theor. {\textbf{119}} (2012)   331 -- 363.

\bibitem{kulish:81}
P.~Kulish, N.~Reshetikhin  and E.~Sklyanin,
\newblock {\em {Yang-Baxter equation and representation theory: I}},
\newblock Lett. Math. Phys. {\textbf{5}} (1981)   393--403.

\bibitem{kulish:82}
P.~P. Kulish and E.~K. Sklyanin,
\newblock {\em {Quantum spectral transform method}},
\newblock Lecture Notes in Physics {\textbf{151}} (1982)   61--119.

\bibitem{kirillov:87}
A.~N. {Kirillov} and N.~Y. {Reshetikhin},
\newblock {\em {Exact solution of the integrable XXZ Heisenberg model with
  arbitrary spin. I. The ground state and the excitation spectrum}},
\newblock J. Phys. A : Math. Gen. {\textbf{20}} (1987)   1565--1585.

\bibitem{yung:95}
C.~M. {Yung} and M.~T. {Batchelor},
\newblock {\em {Exact solution for the spin-s XXZ quantum chain with
  non-diagonal twists}},
\newblock Nucl. Phys. B {\textbf{446}} (1995)   461--484.

\bibitem{faddeev:96}
L.~D. {Faddeev},
\newblock {\em {How Algebraic Bethe Ansatz works for integrable model}},
\newblock arXiv:hep-th/9605187 (1996).

\bibitem{izergin:92}
A.~G. {Izergin}, D.~A. {Coker} and V.~E. {Korepin},
\newblock {\em {Determinant formula for the six-vertex model}},
\newblock J. Phys. A : Math. Gen. {\textbf{25}} (1992)   4315--4334.

\bibitem{elkies:92}
N. Elkies, G. Kuperberg, M. Larsen and J. Propp,
\newblock {\em {Alternating-Sign Matrices and Domino Tilings (Part II)}},
\newblock J. Alg. Comb. {\textbf{1}} (1992)   219--234.

\bibitem{razumov:06}
A.V. Razumov and Yu.~G. Stroganov,
\newblock {\em {Bethe roots and refined enumeration of alternating-sign
  matrices}},
\newblock J. Stat. Mech. (2006)   P07004.

\bibitem{razumov:07}
A.~V. Razumov, Yu.~G. Stroganov  and P.~Zinn-Justin,
\newblock {\em {Polynomial solutions of qKZ equation and ground state of XXZ
  spin chain at $\Delta = -1/2$}},
\newblock J. Phys. A : Math. Gen. {\textbf{40}} (2007)   11827.

\bibitem{hagendorf:14susybc}
C.~ Hagendorf,
\newblock {\em Dynamic lattice supersymmetry in spin chains: singlets and boundary
  conditions},
\newblock in preparation (2014).

\bibitem{hagendorf:14spinreversal}
C.~ Hagendorf,
\newblock {\em The nineteen-vertex model and alternating sign matrices. II.
  Antiperiodic boundary conditions},
\newblock in preparation (2014).

\end{thebibliography}
\end{document}